\documentclass[a4]{amsart}
\usepackage[hmargin=2cm,bmargin=2.5cm,tmargin=3cm]{geometry}

\usepackage[usenames]{color}
\usepackage{graphbox}

\usepackage{wasysym}
	
\usepackage{amsmath,amssymb,amsthm,amsfonts,enumerate,url,mathbbol,mathrsfs,esint,tikz,graphicx,
			pifont,float}
\usepackage{hyperref}
\usetikzlibrary{calc}
\usetikzlibrary{through}

\usepackage[normalem]{ulem}
\usepackage[utf8]{inputenc}

\theoremstyle{plain}
\newtheorem{theorem}{Theorem}[section]
\newtheorem{lemma}[theorem]{Lemma}
\newtheorem{proposition}[theorem]{Proposition}
\newtheorem{corollary}[theorem]{Corollary}

\theoremstyle{definition}
\newtheorem{definition}[theorem]{Definition}
\newtheorem{example}[theorem]{Example}

\newtheorem{remark}[theorem]{Remark}

\numberwithin{equation}{section}
\numberwithin{figure}{section}

\newcommand{\R}{\mathbb{R}}
\newcommand{\N}{\mathbb{N}}

\newcommand\mv{\mathsf{v}} 
\newcommand\mw{\mathsf{w}} 
\newcommand\mz{\mathsf{z}} 
\newcommand\me{\mathsf{e}} 
\newcommand{\VertexSet}{\mathsf{V}} 
\newcommand{\EdgeSet}{\mathsf{E}} 
\newcommand\NeuSet{\mathsf{N}} 
\newcommand{\cardV}{N} 

\newcommand{\Graph}{\mathcal{G}} 

\newcommand{\parti}{{\mathcal P}} 
\newcommand{\nenergy[1]}{{\Lambda}^N_{#1}} 

\newcommand{\ndoptenergy[1]}{{\mathcal L}^{N,D}_{#1}} 

\newcommand{\denergy[1]}{{\Lambda}^D_{#1}} 
\newcommand{\noptenergy[1]}{{\mathcal L}^{N,r}_{#1}} 
\newcommand{\doptenergy[1]}{{\mathcal L}^D_{#1}} 

\newcommand{\noptenergylax[1]}{{\mathcal L}^{N,c}_{#1}} 

\newcommand{\PendTwoConn}{{\mathsf{P}}_2} 

\def\:{\thinspace:\thinspace}

\begin{document}

\title{Asymptotics and estimates for spectral minimal partitions of metric graphs} 

\author[M.~Hofmann]{Matthias Hofmann}
\address{Matthias Hofmann, Grupo de F\'isica Matem\'atica, Faculdade de Ci\^encias, Universidade de Lisboa, Campo Grande, Edif\'icio C6, P-1749-016 Lisboa, Portugal}
\email{mhofmann@fc.ul.pt}

\author[J.~B.~Kennedy]{James B.~Kennedy}
\address{James B.~Kennedy, Grupo de F\'isica Matem\'atica \textit{and} Departamento de Matem\'atica, Faculdade de Ci\^encias, Universidade de Lisboa, Campo Grande, Edif\'icio C6, P-1749-016 Lisboa, Portugal}
\email{jbkennedy@fc.ul.pt}

\author[D.~Mugnolo]{Delio Mugnolo}
\address{Delio Mugnolo, Lehrgebiet Analysis, Fakult\"at Mathematik und Informatik, Fern\-Universit\"at in Hagen, D-58084 Hagen, Germany}
\email{delio.mugnolo@fernuni-hagen.de}

\author[M.~Pl\"umer]{Marvin Pl\"umer}
\address{Marvin Pl\"umer, Lehrgebiet Analysis, Fakult\"at Mathematik und Informatik, Fern\-Universit\"at in Hagen, D-58084 Hagen, Germany}
\email{marvin.pluemer@fernuni-hagen.de}


\thanks{The authors would like to thank Pavel Kurasov and Ji\v{r}\'i Lipovsk\'y for helpful comments on the Weyl asymptotics of quantum graphs, as well as the anonymous referee for their very careful and thoughtful work in helping to improve the paper. The work of M.H.\ and J.B.K.\ was supported by the  Funda\c{c}\~ao para a Ci\^encia e a Tecnologia, Portugal, via the program ``Investigador FCT'', reference IF/01461/2015 (J.B.K.), and via project PTDC/MAT-CAL/4334/2014 (M.H. and J.B.K.). The work of D.M. and M.P.\ was supported by the Deutsche Forschungsgemeinschaft (Grant 397230547). This article is based upon work from COST Action 18232 MAT-DYN-NET, supported by COST (European Cooperation in Science and Technology), \url{www.cost.eu}.} 



\begin{abstract}
We study properties of spectral minimal partitions of metric graphs within the framework recently introduced in [Kennedy {\it et al}, Calc.\ Var.\ \textbf{60} (2021), 61]. We provide sharp lower and upper estimates for minimal partition energies in different classes of partitions; while the lower bounds are reminiscent of the classic isoperimetric inequalities for metric graphs, the upper bounds are more involved and mirror the combinatorial structure of the metric graph as well. Combining them, we deduce that these spectral minimal energies also satisfy a Weyl-type asymptotic law similar to the well-known one for eigenvalues of quantum graph Laplacians with various vertex conditions. Drawing on two examples we show that in general no second term in the asymptotic expansion for minimal partition energies can exist, but show that various kinds of behaviour are possible. We also study certain aspects of the asymptotic behaviour of the minimal partitions themselves.
\end{abstract}

\maketitle

\section{Introduction}
\label{sec:intro}

Spectral minimal partitions were first introduced on planar domains in~\cite{ConTerVer05} and have been a popular topic within spectral theory ever since; we refer the interested reader to the survey~\cite{BNHe17}. Roughly speaking associated with each partition $\parti$ of a domain $\Omega$ into $k$ subdomains $\Omega_1,\ldots,\Omega_k$, one can consider the $p$-mean, $\Lambda_p (\parti)$, $p \in [1,\infty]$, of the $k$-vector of first Dirichlet Laplacian eigenvalues on each subdomain $\Omega_i$. Minimising over all suitable partitions $\parti$ yields the \emph{spectral minimal partitions}, which attain this minimal value. On the one hand, such spectral minimal partitions represent a natural way to partition an object which reflects both its geometric and its metric structure; on the other, their story is closely intertwined with that of the Laplacian eigenvalues of the whole domain $\Omega$, and they are sometimes used as proxies for the eigenvalues of $\Omega$.

For all these reasons, it is also natural to study such partitions on metric graphs. In a recent joint work with Pavel Kurasov and Corentin Léna~\cite{KenKurLen20}, two of the present authors undertook what was perhaps the first systematic such study, although partitions of metric graphs had already been considered in a rather different context in \cite{BanBerRaz12}. In this case, instead of considering subdomains, one partitions a metric graph $\Graph$ into $k$ connected subgraphs $\mathcal G_1,\ldots,\mathcal G_k$, which we will call \emph{clusters}. Minimising the $p$-mean of the $k$-vector of eigenvalues of a suitable Laplacian restricted to each such cluster $\mathcal G_i$ again leads to some number $\mathcal L_{k,p}(\mathcal G)$; as before, the \emph{spectral minimal partitions} are the partitions attaining this minimal value.

However, already here differences emerge. The well-posedness of a number of different spectral minimal partition problems was shown in~\cite{KenKurLen20}. In particular one need not restrict to just Dirichlet eigenvalues on the clusters; on metric graph partitions one may also consider the first nontrivial eigenvalue of the Laplacian with natural (a.k.a.\ standard, Neumann--Kirchhoff) vertex conditions imposed at the cut points. Just as the nodal domains of an eigenfunction on the whole graph give rise to a Dirichlet partition, so do its \emph{Neumann domains} (see in particular \cite{AloBan19,AloBanBerEgg20}) give rise to such a ``Neumann'' partition; unlike on domains, one can now study the corresponding minimisation problem.

There is, likewise, considerable freedom as to how to construct the clusters in the first place: since in principle a partition arises by \emph{cutting} the graph in a finite number of places in the right way, deciding exactly where one is allowed to cut gives rise to several distinct notions of partition. (We will recall the details in Section~\ref{sec:prelim} and an appendix.)

At any rate, several qualitative properties of such partitions were also discussed in \cite{KenKurLen20}; in particular, the close relationship between the minimal partitions and their energies on the one hand, and eigenfunctions and eigenvalues of the Laplacian on the same metric graph on the other, was considered in some detail in~\cite[\S~5 and \S~7]{KenKurLen20}. This and other work (see below) suggest that, even more than on domains, spectral minimal energies of $k$-partitions may be reasonable proxies for the eigenvalues of the Laplacian on the whole graph, in addition to their interpretation as a way to partition a graph into $k$ ``analytically similar'' pieces. It is thus natural to consider their qualitative and quantitative properties, in particular in terms of how these spectral minimal energies depend on ``geometric quantities'' of metric graphs like the total length or the number of vertices of degree one, as well as their asymptotic behaviour for large $k$, a question of some interest on domains (see the discussion in \cite[\S~10.9]{BNHe17} as well as Remark~\ref{rem:hexag} below).

The dependence of Laplacian eigenvalues on the geometric, topological and metric properties of the graph has been intensively studied in recent years; we refer, among others, to the pioneering paper~\cite{Nic87}, as well as \cite{Fri05}, the recent contributions~\cite{KenKurMal16,BanLev17,BerKenKur19}, and the references therein, which seek to estimate the eigenvalues of Laplacians on metric graphs by a combination of metric and combinatorial quantities.

The current work is thus, firstly, devoted to estimating spectral minimal energies in terms of such properties of the graph, similar to eigenvalue estimates. In fact, our lower estimates for spectral minimal energies are mostly easy consequences of known isoperimetric inequalities for eigenvalues of metric graphs. Our upper estimates, on the contrary, require different methods: they are variational in nature, in that they are based on considering suitable test partitions whose energies can be efficiently estimated, in a way reminiscent of the approach in~\cite[\S~4]{BerKenKur17}.

But, secondly, as a notable by-product of our estimates, which are both sharp for each $k$ and asymptotically sharp for each graph as $k \to \infty$, we can obtain asymptotic relations of Weyl type for the spectral minimal energies which strongly recall the eigenvalue Weyl asymptotics. More precisely, we will show that the energies $\mathcal{L}_{k,p}$ grow as
\begin{equation}\label{eq:weyl-gen}
\frac{\pi^2}{|\Graph|^2}k^2 + \mathcal{O}(k)\quad \hbox{as }k\to \infty,
\end{equation}
exactly like the eigenvalues of various realisations of the Laplacian on compact metric graphs \cite{BolEnd09,CuWa05,Nic87,OdzSce19}; this also ties in with the result that, ``generically'' (roughly speaking, for a given graph topology without loops, for ``most'' possible edge lengths), the eigenfunctions of the $k$-th Laplacian eigenvalue of the whole graph, say with natural vertex conditions, have $k + \mathcal{O}(1)$ nodal and Neumann domains (see, e.g., \cite{BanBerRaz12} and \cite{AloBanBerEgg20}, respectively). We can also obtain a description of the corresponding minimal partitions (see below).

We will also show that \eqref{eq:weyl-gen} cannot be improved, that is, that in general there is no second term in the asymptotic expansion. By way of comparison, up until now this idea does not seem to have been formalised for the eigenvalue Weyl asymptotics, although Pavel Kurasov points out (private communication) that \cite[Theorem~2]{Kur08} can be used to show that on graphs with rationally dependent edge lengths, in general such a second term will not exist.

Further parallels between the spectral minimal energies and quantum graph Laplacian eigenvalues will be investigated in a forthcoming paper \cite{HofKen21}, where the focus will also be on interpreting partitions as suitable cuts of a given graph, and the relationship between the \emph{rank} of the cut and properties of the partition. This will also help to clarify the link between the Neumann-type partition problems considered here and the Neumann domains studied in \cite{AloBan19,AloBanBerEgg20}.

Let us sketch the plan of this paper. We recall all relevant definitions in Section~\ref{sec:prelim}, including the relevant notions of partition. Our main results, namely the principal two-sided estimates on the minimal partition energies $\mathcal L_{k,p}$, the asymptotics \eqref{eq:weyl-gen}, and also a result on the asymptotic behaviour of the optimal partitions themselves, are collected in Section~\ref{sec:main}. The lower bounds on $\mathcal L_{k,p}$ require different techniques from the upper bounds. Hence we group all the proofs of the lower bounds, including extensions of the results of Section~\ref{sec:main}, in Section~\ref{sec:estimates-minimal-lower} and the proofs and extensions of the upper bounds in Section~\ref{sec:estimates-minimal-upper}; in each section, we first discuss the case of Dirichlet partitions, which is more instructive and sometimes more delicate, and then give the corresponding results in the natural case. Section~\ref{sec:asymp-behaviour} is devoted to the proof of our main result on the asymptotic behaviour of the spectral minimal partitions (Theorem~\ref{thm:asymptotic-size}), which controls the size of the maximal cluster within the partition, as well as certain consequences of this result (see Section~\ref{sec:main} for details).

Note that there is no general known inequality between the lowest non-zero eigenvalues of the Laplacian with Dirichlet and with natural vertex conditions. Hence it seems that, likewise, no immediate inequality between spectral minimal energies with Dirichlet and natural vertex conditions is available. In particular, this means as yet no interlacing techniques are available, e.g.\ for the purpose of proving asymptotics of energies of spectral minimal $k$-partitions as $k\to \infty$: all our results have to be proved separately for the cases of Dirichlet and natural conditions.

We conclude our note by discussing, in Section~\ref{sec:asymp}, two simple illustrative examples which allow us to show the rich behaviour of the correction term $\mathcal O(k)$ in~\eqref{eq:weyl-gen} and in particular show that the correction term $\mathcal O(k)$ in~\eqref{eq:weyl-gen} will not in general contain any first order term. Rather, if we write $\mathcal{L}_{k,p} = \frac{\pi^2}{L^2}k^2 + c_k k$, $c_k \in \R$, then the set of points of accumulation of the $c_k$ may be a finite set of any cardinality, as our first example of equilateral stars shows (Section~\ref{sec:equilateral}), or an interval, as our second example of two disjoint intervals of incommensurate lengths shows (Section~\ref{sec:irrational}). We also set up the former to give an explicit example for which there is no second term in the Weyl asymptotics at the same time, see Remark~\ref{rem:second-weyl}. It seems reasonable to expect these two types of behaviour to be generic for graphs with rationally dependent and independent edge lengths, respectively, and the same to hold for the Weyl asymptotics of the Laplacian eigenvalues, but it would go well beyond the scope of this note to explore the question further.

For ease of reference, we collect some useful isoperimetric-type inequalities for Laplacian eigenvalues on graphs in an appendix.

\section{Finite quantum graphs and spectral minimal partitions}
\label{sec:prelim}

As mentioned, the present paper is strongly motivated by the setting introduced in~\cite{KenKurLen20}, which we are going to recall and summarise briefly (and every now and then somewhat imprecisely) for the convenience of the reader. We start with some preliminaries on metric graphs.

A  \textit{compact metric graph} $\Graph$ is a finite disjoint union of bounded intervals \((\mathcal I_\me)_{\me\in\EdgeSet}\) connected in a network-like fashion by possibly gluing their endpoints. The set \(\VertexSet\) of glued end points will be referred to as the \emph{vertex set} of \(\Graph\); the set \(\EdgeSet\) will be refered to as the \emph{edge set} of \(\Graph\). The length of an edge \(\me\in\EdgeSet\) -- which is the length of the corresponding interval \(\mathcal I_\me\) -- will be denoted by \(|\me|\). Given a measurable subset \(\mathcal A\subset \Graph\) its Lebesgue measure will be denoted by \(|\mathcal A|\); note that this notation is in line with the notation for the edge lengths in \(\mathcal G\), as each edge \(\me\in\EdgeSet\) will be identified with the subset of \(\Graph\) corresponding to the interval \(\mathcal I_\me\). The \emph{total length} of \(\Graph\) will be denoted by
	\[L:=|\Graph|=\sum_{\me\in\EdgeSet}|\me|.\]
We say $\Graph$ is connected if it is connected as a metric space for the canonical distance function induced by Euclidean distance on each edge and the above construction. Before continuing, we introduce a number of quantities of a given graph $\mathcal{G}$ which will be important in the sequel. We refer to~\cite{Mug19} for a more rigorous definition of \(\Graph\) as a metric measure space.

\begin{definition}
\label{def:various}
Let $\mathcal{G}$ be a compact, connected metric graph.
\begin{enumerate}
\item The longest edge length of $\Graph$ will be denoted by $\ell_{\max} := \max_{\me \in \EdgeSet} |\me|$; the shortest edge length by $\ell_{\min} := \min_{\me \in \EdgeSet} |\me|$.
\item The \emph{Betti number} $\beta \geq 0$ of $\mathcal{G}$ is the number of independent cycles in $\mathcal{G}$; equivalently, $\beta = |\EdgeSet| - |\VertexSet| + 1$.
\item The \emph{girth} $\mathfrak{s} \in (0,\infty]$ is the length of the shortest cycle in $\mathcal{G}$; if $\mathcal{G}$ is a tree, then it is defined to be $\infty$.
\end{enumerate}
\end{definition}

Given such a metric graph $\Graph$, we may consider a ($k$-)partition $\parti$ of $\Graph$ as a family $\parti =\{\mathcal G_1,\ldots,\mathcal G_k\}$ of connected, distinct, metric subgraphs of $\mathcal G$ (its \textit{clusters}) with mutually disjoint interiors and whose union yields $\mathcal G$. This partition may formally be considered to be a collection of $k$ connected components of a \emph{cut} of $\Graph$; we refer to Appendix~\ref{sec:appart} or \cite[\S~2]{KenKurLen20} for the details.

Depending on which kinds of cuts are permitted, one may naturally construct different classes of partitions. In~\cite{KenKurLen20} several such classes were considered, two of which are:
\begin{itemize}
\item the most general class $\mathfrak{C}_k (\Graph)$ of all \emph{connected} exhaustive $k$-partitions of $\Graph$, where we merely insist that each cluster $\Graph_i$ be connected: in particular, we may cut through any points in the interior of $\Graph_i$ as long as it remains connected;
\item the more restrictive class $\mathfrak{R}_k (\Graph)$ of all \emph{rigid} exhaustive $k$-partitions, where we only admit cuts at \emph{separating points}, those points in $\Graph$ on the boundary between clusters, although these points may be cut arbitrarily as long as the $\Graph_i$ remain connected.
\end{itemize}
(Again, see Appendix~\ref{sec:appart}, in particular Definition~\ref{def:classification}, for the precise definitions.) These are, arguably, the most natural classes, as they are the ones that are closed with respect to the natural topology on the space of all partitions discussed in detail in~\cite[\S~3]{KenKurLen20}. The set of all rigid partitions is immediately seen to be a subset of the set of all loose partitions of $\Graph$. The difference between these notions is probably best explained by the example, given in Figures~\ref{fig:basic-example} and~\ref{fig:basic-example-loose} of a lasso graph, with two rigid partitions, and one non-rigid partition.

\begin{figure}[H]
\begin{tikzpicture}[scale=1.2]
\coordinate (a) at (0,0);
\coordinate (b) at (2,0);
\coordinate (c) at (4,0);
\draw[thick] (a) -- (b);
\draw[thick,bend left=90]  (b) edge (c);
\draw[thick,bend right=90]  (b) edge (c);
\draw[fill] (0,0) circle (1.75pt);
\draw[fill] (2,0) circle (1.75pt);
\draw[fill] (4,0) circle (1.75pt);
\node at (b) [anchor=west] {$\mv$};
\node at (a) [anchor=north] {$\mw$};
\node at (c) [anchor=east] {$\mz$};
\node at (1,0) [anchor=south] {$\me_1$};
\node at (3,0.6) [anchor=south west] {$\me_2$};
\node at (3,-0.6) [anchor=north west] {$\me_3$};
\end{tikzpicture}
\caption{The lasso $\Graph$.}\label{fig:basic-example}
\end{figure}
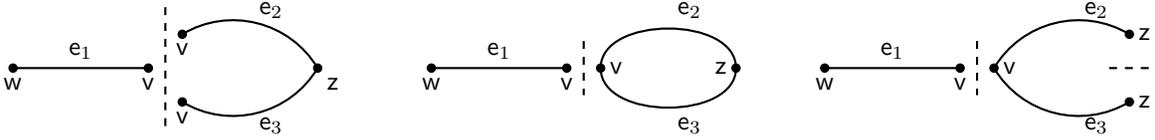
\begin{figure}[H]
\centering
\begin{minipage}{.3\textwidth}
\centering
\begin{tikzpicture}[scale=0.9]
\coordinate (e) at (6.5,0);
\coordinate (f) at (8.5,0);
\coordinate (g) at (9,0.5);
\coordinate (h) at (9,-0.5);
\coordinate (i) at (11,0);
\draw[thick] (e) -- (f);
\draw[thick,bend right=45] (i) edge (g);
\draw[thick,bend left=45] (i) edge (h);
\draw[fill] (6.5,0) circle (1.75pt);
\draw[fill] (8.5,0) circle (1.75pt);
\draw[fill] (9,0.5) circle (1.75pt);
\draw[fill] (9,-0.5) circle (1.75pt);
\draw[fill] (11,0) circle (1.75pt);
\node at (e) [anchor=north] {$\mw$};
\node at (f) [anchor=north] {$\mv$};
\node at (g) [anchor=north] {$\mv$};
\node at (h) [anchor=north] {$\mv$};
\node at (i) [anchor=north west] {$\mz$};
\node at (7.5,0) [anchor=south] {$\me_1$};
\node at (10,0.6) [anchor=south west] {$\me_2$};
\node at (10,-0.6) [anchor=north west] {$\me_3$};
\draw[thick,dashed] (8.75,0.9) -- (8.75,-0.9);
\end{tikzpicture}
\end{minipage}
\begin{minipage}{.3\textwidth}
\centering
\begin{tikzpicture}[scale=0.9]
\coordinate (a) at (-.5,0);
\coordinate (b1) at (1.5,0);
\coordinate (b) at (2,0);
\coordinate (c) at (4,0);
\draw[thick] (a) -- (b1);
\draw[thick,bend left=90]  (b) edge (c);
\draw[thick,bend right=90]  (b) edge (c);
\draw[fill] (a) circle (1.75pt);
\draw[fill] (b) circle (1.75pt);
\draw[fill] (b1) circle (1.75pt);
\draw[fill] (c) circle (1.75pt);
\node at (b) [anchor=west] {$\mv$};
\node at (b1) [anchor=north] {$\mv$};
\node at (a) [anchor=north] {$\mw$};
\node at (c) [anchor=east] {$\mz$};
\draw[thick,dashed] (1.75,0.4) -- (1.75,-0.4);
\node at (0.5,0) [anchor=south] {$\me_1$};
\node at (3,0.6) [anchor=south west] {$\me_2$};
\node at (3,-0.6) [anchor=north west] {$\me_3$};
\end{tikzpicture}
\end{minipage}
\begin{minipage}{.3\textwidth}
\centering
\begin{tikzpicture}[scale=0.9]
\coordinate (a) at (0,0);
\coordinate (b) at (2,0);
\coordinate (c) at (2.5,0);
\coordinate (d1) at (4.5,-.5);
\coordinate (d2) at (4.5,.5);
\draw[thick] (a) -- (b);
\draw[thick,bend right=45] (c) edge (d1);
\draw[thick,bend left=45] (c) edge (d2);
\draw[fill] (a) circle (1.75pt);
\draw[fill] (b) circle (1.75pt);
\draw[fill] (c) circle (1.75pt);
\draw[fill] (d1) circle (1.75pt);
\draw[fill] (d2) circle (1.75pt);
\node at (a) [anchor=north] {$\mw$};
\node at (b) [anchor=north] {$\mv$};
\node at (c) [anchor=west] {$\mv$};
\node at (d1) [anchor=west] {$\mz$};
\node at (d2) [anchor=west] {$\mz$};
\node at (1,0) [anchor=south] {$\me_1$};
\node at (3.7,0.6) [anchor=south west] {$\me_2$};
\node at (3.7,-0.6) [anchor=north west] {$\me_3$};
\draw[thick,dashed] (2.25,0.4) -- (2.25,-0.4);
\draw[thick,dashed] (4.2,0) -- (4.8,0);
\end{tikzpicture}
\end{minipage}
\caption{Left and centre: two different rigid 2-partitions of $\Graph$ (the only separating point is $\mv$); right: a connected 2-partition of $\Graph$ which is not rigid (the only separating point is $\mv$ but we are additionally cutting through $\mz$; observe that additionally cutting through $\mv$ would not be admissible, as this would yield a 3-partition).}\label{fig:basic-example-loose}
\end{figure}

Given a metric subgraph $\mathcal{H}$ of $\mathcal{G}$, in practice a cluster of a partition, we can consider the function spaces $C(\mathcal{H})$, $L^2(\mathcal{H})$ and $H^1(\mathcal{H})$ defined on it in the usual way; given a subset $\VertexSet_0$ of $\VertexSet (\mathcal{H})$, $H^1_0(\mathcal{H}; \VertexSet_0)$ is the ideal of $H^1(\mathcal{H})$ consisting of all $H^1(\mathcal G)$-functions vanishing at $\VertexSet_0$. We then define the quadratic Dirichlet form
\begin{displaymath}
	a(f) := \int_{\mathcal{H}} |f'(x)|^2\, \textrm{d}x
\end{displaymath}
on the domain $H^1 (\mathcal{H})$ or $H^1_0 (\mathcal{H}; \VertexSet_0)$. In the former case, the associated operator is the Laplacian with so-called \emph{standard} or \emph{natural} vertex conditions; the functions in its domain are in $C(\mathcal{H}) \cap L^2(\mathcal{H})$ and are edgewise $H^2$, and satisfy a Kirchhoff condition; in the second case, the functions additionally satisfy a Dirichlet (zero) condition at every vertex in $\VertexSet_0$. Such Laplacians defined on metric graphs are usually known as \textit{quantum graphs}. Since $\mathcal{H} \subseteq \mathcal{G}$ is a compact metric graph, such Laplacians are self-adjoint operators with compact resolvent and in particular have discrete, real spectrum. We will be interested in their respective smallest nontrivial eigenvalues, which may be described variationally by
\begin{equation}
\label{eq:mu2}
	\mu_2 (\mathcal{H}) = \inf \left\{ \frac{\int_{\mathcal{H}} |f'(x)|^2\, \textrm{d}x}{\int_{\mathcal{H}} |f(x)|^2\, \textrm{d}x} :
		0 \neq f \in H^1 (\mathcal{H}) \text{ and } \int_\mathcal{H} f(x)\,\textrm{d}x = 0\right\}
\end{equation}
in the case of the Laplacian with standard vertex conditions, and
\begin{equation}
\label{eq:lambda1}
	\lambda_1 (\mathcal{H}) = \inf \left\{ \frac{\int_{\mathcal{H}} |f'(x)|^2\, \textrm{d}x}{\int_{\mathcal{H}} |f(x)|^2\, \textrm{d}x} :
		0 \neq f \in H^1_0 (\mathcal{H};\VertexSet_0) \right\}
\end{equation}
for the Laplacian with at least one Dirichlet condition, i.e., if $\VertexSet_0 \neq \emptyset$. Equality in each case is achieved exactly when $f$ is a corresponding eigenfunction. These eigenvalues may be shown to be strictly positive if $\mathcal{H}$ is connected.

Now given a $k$-partition $\mathcal P=\{\mathcal G_1,\ldots,\mathcal G_k\}$, since each cluster $\mathcal G_i$ is compact and connected, $\lambda_1 (\mathcal{G}_i)$ and $\mu_2 (\mathcal{G}_i)$ are well defined and strictly positive, where in the former case the Dirichlet conditions are taken at the cut points of $\Graph_i$. We can thus consider the following partition \textit{energies}
\begin{equation}
\label{eq:nenergy}
	\nenergy[p] (\parti) = \begin{cases} \left(\frac{1}{k}\sum\limits_{i=1}^k \mu_2(\Graph_i)^p\right)^{1/p}
	\qquad &\text{if } p \in [1,\infty),\\ \max\limits_{i=1,\ldots,k} \mu_2(\Graph_i) 
	\qquad &\text{if } p = \infty, \end{cases}
\end{equation}
and
\begin{equation}
\label{eq:denergy}
	\denergy[p] (\parti) = \begin{cases} \left(\frac{1}{k}\sum\limits_{i=1}^k \lambda_1(\Graph_i)^p\right)^{1/p}
	\qquad &\text{if } p \in [1,\infty),\\ \max\limits_{i=1,\ldots,k} \lambda_1(\Graph_i) 
	\qquad &\text{if } p = \infty, \end{cases}
\end{equation}
respectively. We can finally consider 
\[
\noptenergy[k,p](\Graph) = \min_{\parti \in \mathfrak{R}_k (\Graph)} \nenergy[p](\parti)\quad\hbox{and}\quad \noptenergylax[k,p](\Graph) = \min_{\parti \in \mathfrak{C}_k (\Graph)} \nenergy[p](\parti), \]
the minimum of $\nenergy[p] (\parti)$ over all rigid/connected $k$-partitions, respectively; and
\[
\doptenergy[k,p](\Graph) = \min_{\parti \in \mathfrak{R}_k (\Graph)} \denergy[p](\parti) = \min_{\parti \in \mathfrak{C}_k (\Graph)} \denergy[p](\parti),
\]
the minimum of $\denergy[p] (\parti)$ over all rigid $k$-partitions. It was proved in \cite[Corollary~4.8]{KenKurLen20} that under our assumptions on $\Graph$ all these problems do indeed admit minima (by~\cite[Lemma~4.3]{KenKurLen20}, and that it does indeed make no difference to study the minimum $\denergy[p] (\parti)$ over all rigid or all connected $k$-partitions. We generically refer to all these minima as \textit{spectral minimal energies} of $\mathcal G$ and, as mentioned, the corresponding minimising partitions as \textit{spectral minimal partitions}.

\section{Main results: asymptotic behaviour of the optimal energies and partitions}
\label{sec:main}

We start by summarising our principal results, which give concrete two-sided bounds on the quantities $\doptenergy[k,p] (\Graph)$, $\noptenergy[k,p] (\Graph)$ and $\noptenergylax[k,p] (\Graph)$, and as a consequence describe their asymptotic behaviour. Actually, we can say more, both about the asymptotic behaviour of the clusters of the optimal partitions, and in terms of concrete two-sided bounds on these quantities for finite $k$. The compact, connected metric graph $\Graph$ will be fixed throughout, and we recall, in addition to the notation from Definition~\ref{def:various}, that $\Graph$ is taken to have $|\EdgeSet| \geq 1$ edges, total length $L$, and $|\NeuSet|$ vertices of degree one.

\begin{theorem}\label{thm:dirichlet-asymptotics-first-term}
Let $p \in [1,\infty]$. Then
\begin{displaymath}
	\frac{\pi^2}{4k L^2} \left ( k^3 + 3  (k-\beta -|\NeuSet|)^3 \right ) \leq \doptenergy[k,p] (\Graph) \leq \frac{\pi^2}{L^2}\left(k+\left(|\EdgeSet|-1-\left\lfloor\frac{|\NeuSet|}{2}\right\rfloor\right)\right)^2
\end{displaymath}
for all sufficiently large $k \geq 2$, in particular for
\begin{displaymath}
	k \geq \max \left\{\beta+|\NeuSet|, \frac{L}{\ell_{\min}} + |\EdgeSet| - 1\right\}.
\end{displaymath}
In particular,
\begin{equation}
\label{eq:dirichlet-weyl}
	\doptenergy[k,p] (\Graph) = \frac{\pi^2}{L^2}k^2 + \mathcal{O}(k) \qquad \text{as } k \to \infty.
\end{equation}
\end{theorem}

This theorem will be an immediate consequence of the results of Sections~\ref{sec:dirichlet-lower} and~\ref{sec:dirichlet-upper}; see in particular Theorems~\ref{thm:improveddirichletestimate} and~\ref{thm:dopt-upper-bound}. Actually, we can give slightly sharper (but often more involved) lower bounds in some cases; in addition to Theorem~\ref{thm:improveddirichletestimate} we mention Corollary~\ref{cor:improveddirichletestimate-largek}.

\begin{theorem}
\label{thm:neumann-asymptotics-first-term}
Let $p \in [1,\infty]$. Then
\begin{equation}
\label{eq:neumann-two-sided}
	\frac{\pi^2}{L^2}k^2 \leq \noptenergylax[k,p] (\Graph) \leq \noptenergy[k,p] (\Graph) \leq \frac{\pi^2}{L^2}\big(k + (|\EdgeSet|-1)\big)^2.
\end{equation}
for all $k \geq 1$ in the case of the lower bound, and for all sufficiently large $k$ in the case of the upper bound, in particular for $k \geq 5|\EdgeSet|-1$. In particular,
\begin{equation}
\label{eq:neumann-weyl}
	\noptenergylax[k,p](\Graph),\,\noptenergy[k,p] (\Graph) = \frac{\pi^2}{L^2} k^2 + \mathcal{O} (k) \qquad \text{as } k \to \infty.
\end{equation}
\end{theorem}

This theorem follows from results in Sections~\ref{sec:neumann-lower} and~\ref{sec:neumann-upper}, in particular Theorems~\ref{th:nopt-lower-bound} and~\ref{thm:nopt-upper-bound} (the latter in conjunction with Remark~\ref{rem:nopt-upper-bound}). In this case, it is possible to say a fair amount about when there is equality in the lower bound in \eqref{eq:neumann-two-sided}; see Propositions~\ref{prop:nopt-lower-bound-equality} and~\ref{prop:nopt-subsequence}.

We can also give a description of the asymptotic behaviour of the minimal partitions realising $\doptenergy[k,p]$, $\noptenergy[k,p]$ etc. It is perhaps not surprising that for $k$ large enough all clusters become either intervals or stars, just as is the case for both the nodal and the Neumann domains of the $k$-th eigenfunction of the Laplacian on the whole graph, see \cite[Proposition~7.4]{AloBanBerEgg20}. Our main result states that in fact, for any $p \in [1,\infty]$, asymptotically all clusters are of length of order $1/k$: no clusters can remain too ``large''.

\begin{theorem}
\label{thm:asymptotic-size}
Fix $p \in [1,\infty]$ and, for each $k \geq 2$, let $\parti_k^N$, $\widetilde\parti_k^N$ and $\parti_k^D$ be any admissible partitions realising $\noptenergy[k,p](\Graph)$, $\noptenergylax[k,p] (\Graph)$ and $\doptenergy[k,p] (\Graph)$, respectively. Denote the size of the largest cluster of each by $L_{\max}^{N,r}(k)$, ${L}_{\max}^{N,c}(k)$ and $L_{\max}^D(k)$, respectively. Then
\begin{equation}
\label{eq:asymptotic-size}
	L_{\max}^{N,r}(k),\, {L}_{\max}^{N,c}(k),\, L_{\max}^D(k) = \mathcal{O}(k^{-1}) \qquad \text{as } k \to \infty.
\end{equation}
\end{theorem}

\begin{remark}\label{rem:hexag}
One of the main open problems in the theory of spectral minimal partitions for planar domains $\Omega$ is the so-called \textit{hexagonal conjecture} that seems to go back to Caffarelli and Lin, see~\cite[\S~10.9.1]{BNHe17}, which postulates that
\begin{equation}
\label{eq:hex-con}
\lim_{k\to \infty}\frac{\doptenergy[k,p]}{k}=\frac{\lambda}{|\Omega|},
\end{equation}
where $\lambda$ is the lowest eigenvalue of the Dirichlet Laplacian on a regular hexagon of unit area (regular hexagons being the tesselating planar domains with minimal first Dirichlet eigenvalue). Of course, on graphs, the geometric side of this question disappears: the correct counterparts of hexagons are just intervals. However, Theorems \ref{thm:dirichlet-asymptotics-first-term} and \ref{thm:asymptotic-size} still cover the natural \emph{analytic} counterpart of \eqref{eq:hex-con}, that
\[
\lim_{k\to \infty}\frac{\doptenergy[k,p]}{k^2}=\frac{\pi^2}{L^2},
\]
including the ``balancing'' statement that in the limit the size of the clusters in the optimal partitions becomes uniform, for every fixed $p \in [1,\infty]$.
\end{remark}

Due to parallels between the respective proofs in the Dirichlet and natural cases, we will group the lower bounds together in Section~\ref{sec:estimates-minimal-lower} and the upper bounds in Section~\ref{sec:estimates-minimal-upper}; the proof of Theorem~\ref{thm:asymptotic-size} will be given in Section~\ref{sec:asymp-behaviour}, where we also collect a couple of results (improved bounds, Corollary~\ref{cor:improveddirichletestimate-largek}, and a monotonicity statement for $\noptenergy[k,p]$ as a function of $k$ for $k$ sufficiently large, Theorem~\ref{thm:neumann-monotonicity}) which follow from Theorem~\ref{thm:asymptotic-size}. We also show that this monotonicity result does not necessarily hold for all $k$, see Example~\ref{ex:nomon}. Finally, we recall that Section~\ref{sec:asymp} is devoted to the non-existence of a second term (i.e., term of first order) in the asymptotic expansions \eqref{eq:dirichlet-weyl} and \eqref{eq:neumann-weyl}. We also set up one of our examples to give an example that there need not be any second term in the Weyl asymptotics for $\mu_k$ (see Remark~\ref{rem:second-weyl}).

\section{Lower bounds}
\label{sec:estimates-minimal-lower}

\subsection{Dirichlet partitions} 
\label{sec:dirichlet-lower}
We first consider lower bounds on the optimal Dirichlet partition energy $\doptenergy[k,p] (\Graph)$.

\begin{theorem}\label{th:dopt-general-lower-bound}
Let $\Graph$ be a compact and connected metric graph with total length $L>0$. For any $p \in [1,\infty]$ and any $k\geq 2$, we have
\begin{equation}
\label{eq:dopt-general-lower-bound}
	\doptenergy[k,p] (\Graph) \geq \frac{\pi^2k^2}{4L^2}.
\end{equation}
Equality implies that $\Graph$ is an equilateral $k$-star $\mathcal{S}_k$.
\end{theorem}

Observe that the special case of $p=\infty$ can also be obtained from combining~\cite[Prop.~5.5]{KenKurLen20} and~\cite[Thm.~1]{Fri05}.

\begin{proof}
Since $\doptenergy[k,p] (\Graph)$ is monotonically increasing in $p \in [1,\infty]$ (see \cite[Prop.~7.1]{KenKurLen20}), it suffices to prove \eqref{eq:nopt-lower-bound} for $p=1$ only.  We suppose that $\Graph_1,\ldots,\Graph_k$ are the clusters of an optimal partition associated with $\doptenergy[k,1] (\Graph)$; then since each has at least one Dirichlet vertex, we may apply the version of Nicaise' inequality for Dirichlet problems cf.\ Proposition~\ref{thm:nicaise} to obtain $\lambda_1 (\Graph_i) \geq \pi^2/(4|\mathcal{G}_i|^2)$, $i=1,\ldots,k$. Thus, by
Jensen's inequality in discrete form applied to the convex map $x \mapsto x^{-2}$, $x>0$, we find
\begin{displaymath}
	\doptenergy[k,1] (\Graph) = \frac{1}{k}\sum_{i=1}^k \lambda_1 (\Graph_i) \geq \frac{\pi^2}{4}\left( \frac{1}{k}\sum_{i=1}^k |\Graph_i|^{-2}
	\right) \geq \frac{\pi^2 k^2}{4L^2}.
\end{displaymath}
This proves \eqref{eq:dopt-general-lower-bound}. For the case of equality, first note that there is equality in Proposition~\ref{thm:nicaise}.(1) if and only if $\Graph_i$ is an interval of length $|\Graph_i|$, with one Dirichlet and one Neumann endpoint (i.e., vertex); this is an immediate consequence of \cite[Lemma~3]{Fri05} together with the variational characterisation of $\lambda_1$. Moreover, equality in Jensen's inequality implies that $|\Graph_1| = \ldots = |\Graph_k| = L/k$. Hence equality in \eqref{eq:dopt-general-lower-bound} (for any $p\geq 1$ and any $k\geq 2$) is only possible if all the $\Graph_i$ are intervals of length $L/k$ with one Dirichlet and one Neumann endpoint. Since the boundary between neighbouring clusters is always marked by a Dirichlet vertex, the only possible connected metric graph that can have these graphs as partition clusters is $\mathcal{S}_k$.
\end{proof}

\begin{remark}\label{rem:dopt-general-lower-bound}
The theorem contains the statement that the optimal $k$-partition of an equilateral $k$-star $\mathcal{S}_k$, for any $p \in [1,\infty]$, is the expected one, i.e., where each edge is a cluster. More interestingly, this partition reflects the nodal pattern of $\lambda_k (\mathcal{S}_k)$; and $\mathcal{S}_k$ is also the (unique) minimiser of $\lambda_k (\Graph)$ among all graphs of fixed total length, as proved by Friedlander \cite{Fri05}. As with Friedlander's inequality, Theorem~\ref{th:dopt-general-lower-bound} implies in particular that the minimal possible values for $\doptenergy[k,p] (\Graph)$ (among all possible graphs $\Graph$ of given length $L$) do not exhibit the asymptotic behaviour $\pi^2k^2/L^2$ which would be consistent with the Weyl asymptotics of each fixed graph.
\end{remark}

In both cases, the divergence from the Weyl asymptotics is due to the factor of 1/4 appearing in Nicaise' inequality for $\lambda_1$, which reflects the case of the interval with only one Dirichlet endpoint. To recover the asymptotically correct value, there needs to be a reasonable ``distribution'' of Dirichlet vertices in the graph; in particular, an improved inequality can only be valid for sufficiently large $k$ or for special classes of graphs. Before stating our improved estimates, we recall that a connected metric graph is called \emph{doubly connected} if it is not simply connected as a metric space, i.e., if at least two edges need to be deleted in order to make it disconnected. We refer to Section \ref{sec:equilateral} for a detailed discussion of the asymptotics for equilateral stars.

\begin{definition}
\label{def:doubly-connected-pendant}
Let \(\mathcal G\) be a compact and connected metric graph. We will call a metric subgraph \(\mathcal G'\subset \mathcal G\) a \emph{doubly connected pendant} of \(\mathcal G\) if \(\mathcal G'\) has non-empty interior, \(\mathcal G'\) is doubly connected and there is exactly one edge \(\me\in\mathcal E\) of strictly positive length connecting \(\mathcal G'\) with its complement $\Graph \setminus \Graph'$. The set of all doubly connected pendants of \(\mathcal G\) will be denoted by \(\PendTwoConn\).
\end{definition}

\begin{example}\label{exa:pc2}
Note that Definition~\ref{def:doubly-connected-pendant} explicitly requires the existence of a \emph{bridge} (of positive length) as a precondition for the existence of any doubly connected pendants. A dumbbell graph (with non-degenerate handle) has two doubly connected pendants, consisting of its two loops. More generally, an {$(m,1,m)$}-pumpkin chain (see~\cite[\S~5]{BerKenKur19}) has, for $m>1$, two doubly connected pendants (the two $m$-pumpkins) but Betti number $2(m-1)$. However, figure-eight graphs and, more generally, flower graphs -- indeed, all doubly connected graphs -- have \emph{none}.
\end{example}

Note that any two distinct doubly connected pendants are disjoint, and that necessarily the Betti number satisfies $\beta \geq |\PendTwoConn|$, as any cycles belonging to different doubly connected pendants are necessarily independent.

\begin{theorem}\label{thm:improveddirichletestimate}
Let $\Graph$ be a compact and connected metric graph with total length $L>0$, $|\NeuSet|$ vertices of degree one and \(|\PendTwoConn|\) doubly connected pendants. Fix $k\ge 2$ and $p\in [1,\infty]$. Then for any $k\ge |\NeuSet|+|\PendTwoConn|$
 we have
\begin{equation}
\label{eq:improveddirichletestimate}
	\doptenergy[k,p] (\Graph) \ge \frac{\pi^2}{4k L^2} \left ( k^3 + 3  (k-|\NeuSet|-|\PendTwoConn|)^3 \right ).
\end{equation}
\end{theorem}

The lower bound in Theorem~\ref{thm:dirichlet-asymptotics-first-term} is an immediate consequence of \eqref{eq:improveddirichletestimate} and the fact that $\beta \leq |\PendTwoConn|$. Also observe that if $\Graph$ is itself doubly connected, then $|\NeuSet|=|\PendTwoConn|=0$, whence \eqref{eq:improveddirichletestimate} reduces to
\begin{equation}\label{eq:dopt-lower-bound-edgeconnectivity}
	\doptenergy[k,p] (\Graph) \geq \frac{\pi^2 k^2}{L^2}
\end{equation}
for all $p \in [1,\infty]$ and $k\geq 2$.

The estimate \eqref{eq:improveddirichletestimate} is asymptotically sharp, in the sense that for any value of $p,|\NeuSet|,|\PendTwoConn|$ there exists a value of $k$ and a family of graphs $\Graph_\varepsilon$ such that, for these values of $p,|\NeuSet|,|\PendTwoConn|,k$ there is equality in \eqref{eq:improveddirichletestimate} as $\varepsilon \to 0$; see Remark~\ref{rem:dirichlet-sharp}. The somewhat complicated case of equality in \eqref{eq:improveddirichletestimate} is discussed in Remark~\ref{rem:all-pendants-are-created-equal}.

\begin{proof}[Proof of Theorem~\ref{thm:improveddirichletestimate}]
Without loss of generality, we may assume that \(k>|\NeuSet|+|\PendTwoConn|\), since \eqref{eq:improveddirichletestimate} reduces to \eqref{eq:dopt-general-lower-bound} for \(k=|\NeuSet|+|\PendTwoConn|\). Firstly, as before, by monotonicity it is sufficient to show \eqref{eq:improveddirichletestimate} for $p=1$. So suppose that $\parti=\{\Graph_1, \ldots, \Graph_k\}$ is an optimal $k$-partition of $\Graph$ for $\doptenergy[k,1] (\Graph)$; then at most $|\NeuSet|$ clusters of $\parti$ can contain a vertex of degree $1$ and at most $|\PendTwoConn|$ clusters can contain a doubly connected pendant of \(\Graph\). Suppose
\begin{displaymath}\label{eq:pendtwoestimatehere}
	j_k \le |\NeuSet|+|\PendTwoConn|<k
\end{displaymath}
of the clusters admit at least one vertex of degree $1$ or contain a doubly connected pendant; then after a renumbering if necessary we may assume that $\Graph_{j_k+1}, \ldots, \Graph_k$ contain neither a vertex of degree $1$ of $\Graph$ nor a doubly connected pendant of \(\mathcal G\): in particular, each \(\Graph_i\) for \(i>j_k\) has at least two boundary vertices that are thus equipped with a Dirichlet condition, and the graph obtained by merging all these vertices of degree \(1\) is doubly connected. Therefore, Proposition~\ref{thm:nicaise}.(2) is applicable to these clusters, yielding \(\lambda_1(\Graph_i)\geq \pi^2/|\mathcal G_i|^2\) for \(i>j_k\). Now, define
\begin{equation}
L_k:= \sum_{i=1}^{j_k} |\mathcal G_i|
\end{equation}
and note that \(L_k<L\) holds, since \(j_k<k\). Then, applying Proposition~\ref{thm:nicaise}.(1) to the other clusters and using Jensen's inequality as in the proof of Theorem~\ref{th:dopt-general-lower-bound}, we see that
\begin{equation}\label{eq:estimates-to-prove-improved-lower-dopt-bound}
\begin{split}
	\doptenergy[k,1] (\Graph) = \denergy[1](\parti) 
	&= \frac{\sum_{i=1}^{j_k} 4\lambda_1(\Graph_i)+ \sum_{i=j_k+1}^k \lambda_1(\Graph_i)}{4 k}
		\\
		&\qquad + \frac{3(k-j_k)}{4k}\frac{1}{k-j_k}\sum_{i=j_k+1}^{k} \lambda_1(\Graph_i)\\
	&\ge \frac{1}{4k}\sum_{i=1}^k \frac{\pi^2}{|\Graph_i|^2} + \frac{3(k-j_k)}{4k} \frac{1}{k-j_k} \sum_{i=j_k+1}^k \frac{\pi^2}{|\Graph_i|^2}\\
	&\ge \frac{1}{4} \frac{\pi^2 k^2}{L^2} + \frac{3(k-j_k)}{4k} \frac{\pi^2 (k-j_k)^2}{{(L-L_k)}^2}\\
	&\ge \frac{1}{4} \frac{\pi^2 k^2}{L^2} + \frac{3(k-j_k)}{4k} \frac{\pi^2 (k-j_k)^2}{{L}^2}\\
	&= \frac{\pi^2}{4k L^2} \left ( k^3 + 3  (k-j_k)^3 \right )\\
	&\ge  \frac{\pi^2}{4k L^2} \left ( k^3 + 3  (k-|\NeuSet|-|\PendTwoConn|)^3 \right ).
\end{split}
\end{equation}
This proves the claim.
\end{proof}

\begin{remark}
\label{rem:all-pendants-are-created-equal}
Let us briefly discuss the cases of equality in \eqref{eq:improveddirichletestimate}. We have already seen in Theorem \ref{th:dopt-general-lower-bound} that equality holds for \(k=|\NeuSet|+|\PendTwoConn|\) if and only if \(\mathcal G\) is the equilateral \(k\)-star. In the case \(k>|\NeuSet|+|\PendTwoConn|\) we need to analyse the estimates in \eqref{eq:estimates-to-prove-improved-lower-dopt-bound}. First of all, note that in this case $\doptenergy[k,p](\Graph) = \doptenergy[k,1](\Graph)$. Now the equalities in the fourth and sixth steps of \eqref{eq:estimates-to-prove-improved-lower-dopt-bound} imply \(L_k=0\) and \(|\NeuSet|+|\PendTwoConn|=j_k=0\). Moreover, equality in Jensen's inequality in the third step yields \(|\Graph_i|=\frac{L}{k}\) for \(i=1,\ldots,k\). Finally, equality in the second step, i.e., in Proposition~\ref{thm:nicaise}.(2), implies that every cluster \(\mathcal G_i\) of an optimal \(k\)-partition \(\mathcal P\) is a \textit{caterpillar graph}, i.e.\ a \(2\)-regular pumpkin chain of length \(\frac{L}{k}\) where one of the two end points (of degree two) is equipped with Dirichlet conditions, see also Figure~\ref{fig:caterpillar-graph}. Therefore, equality in \eqref{eq:improveddirichletestimate} holds for \(k>|\NeuSet|+|\PendTwoConn|\) if and only if \(\Graph\) is obtained by arbitrarily gluing a collection of caterpillar graphs at their Dirichlet vertices so that \(\Graph\) has no vertices of degree one -- in particular \(\mathcal G\) has to be doubly connected.
\end{remark}

\begin{remark}
\label{rem:dirichlet-sharp}
Also note that \eqref{eq:improveddirichletestimate} is asymptotically sharp if \(|\PendTwoConn|>0\) and \(k=|\PendTwoConn|+|\NeuSet|\), in the sense that there exists a family of graphs $\Graph_\varepsilon$ differing only by their edge lengths, for which there is equality in the limit as $\varepsilon \to 0$. To see this consider, for \(m\geq 1\) and \(n\geq 0\), an equilateral \(m+n\)-star graph where \(m\) of the degree one vertices are replaced with a loop of sufficiently small length \(\varepsilon>0\); when \(n=0\) these are the graphs considered in~\cite{KurSer18}. The graph $\mathcal W_{m,n}$ thus obtained has $|\NeuSet|=n$ vertices of degree one and $|\PendTwoConn|=m$ doubly connected pendants. One can show that for \(k=m+n\) an optimal \(k\)-partition for \(\doptenergy[k,p](\mathcal W_{m,n})\) is obtained by cutting through the centre vertex, i.e., it consists of \(m\) lasso graphs and \(n\) intervals with one Neumann and one Dirichlet vertex. For these graphs and $k=m+n$, the right-hand side of \eqref{eq:improveddirichletestimate} is just $\frac{\pi^2k^2}{4L^2}$, corresponding to the optimal energy $\doptenergy[m+n,p]$ of the equilateral $m+n$-star of total length $L$. If in $\mathcal{W}_{m,n}$ we let the length of the loops tend to zero, then stability of $\lambda_1$ with respect to this operation (see \cite{BeLaSu18}) implies that \(\doptenergy[k,p](\mathcal W_{m,n})\) indeed converges to the right-hand side of \eqref{eq:improveddirichletestimate}.
\end{remark}
\begin{figure}[h]
\begin{minipage}{0.49\textwidth}
\centering
\begin{tikzpicture}
\coordinate (a) at (1,0);
\coordinate (b) at (2,0);
\coordinate (c) at (4,0);
\coordinate (d) at (5.5,0);
\coordinate (e1) at (6.7,0.5);
\coordinate (e2) at (6.7,-0.5);
\draw[thick,bend left=70]  (a) edge (b);
\draw[thick,bend right=70]  (a) edge (b);
\draw[thick,bend left=70]  (b) edge (c);
\draw[thick,bend right=70]  (b) edge (c);
\draw[thick,bend left=70]  (c) edge (d);
\draw[thick,bend right=70]  (c) edge (d);
\draw[thick] (d)--(e1);
\draw[thick] (d)--(e2);
\draw[fill] (a) circle (1.75pt);
\draw[fill] (b) circle (1.75pt);
\draw[fill] (c) circle (1.75pt);
\draw[fill] (d) circle (1.75pt);
\draw[fill=white] (e1) circle (1.75pt);
\draw[fill=white] (e2) circle (1.75pt);
\end{tikzpicture}
\caption{A caterpillar graph with Dirichlet vertices marked in white.}\label{fig:caterpillar-graph}
\end{minipage}
\begin{minipage}{0.49\textwidth}
\centering
\begin{tikzpicture}
	\foreach \i in {1,2,3,4,5,6} {
		\coordinate (w\i) at (60*\i+30:1);
		\draw[thick] (0,0) -- (w\i);
		}
	\foreach \i in {5,6} {
		\coordinate (v\i) at (60*\i+30:1.4);
		\draw[thick] (w\i) to [out=60*\i,in=60*\i+300] (v\i);
		\draw[thick] (w\i) to [out=60*\i+60,in=60*\i+120] (v\i);
	}
		\draw[fill] (0,0) circle (1.75pt);
		\foreach \i in {1,2,3,4,5,6} {
		\coordinate (w\i) at (60*\i+30:1);
		\draw[fill] (w\i) circle (1.75pt);
		}
\end{tikzpicture}
\caption{The graph \(\mathcal W_{m,n}\) with \(m=2\) and \(n=4\).}\label{fig:mixed-star-and-windmill-graph}
\end{minipage}
\end{figure}

\subsection{Neumann partitions}
\label{sec:neumann-lower}
We start with an analogue of Theorem~\ref{th:dopt-general-lower-bound} for Neumann partitions. In comparison with the Dirichlet case, providing a complete description of the graphs for which there is equality seems to be a rather difficult problem.

\begin{theorem}\label{th:nopt-lower-bound}
Let $\Graph$ be a compact and connected metric graph with total length $L>0$. For any $p \in [1,\infty]$ and any $k\geq 1$, we have
\begin{equation}
\label{eq:nopt-lower-bound}
	\noptenergy[k,p] (\Graph) \geq  \noptenergylax[k,p] (\Graph) \geq \frac{\pi^2k^2}{L^2}.
\end{equation}
If $\Graph$ is not a loop or if $k\geq 2$, then there is equality if and only if there exists a rigid (respectively, a connected) $k$-partition whose every cluster is an interval of length $L/k$.
\end{theorem}

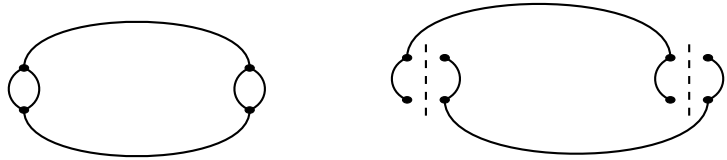
\begin{figure}[H]
\begin{tikzpicture}[yscale=0.7]
\coordinate (a) at (0,0.4);
\coordinate (b) at (3,0.4);
\coordinate (c) at (0,-0.4);
\coordinate (d) at (3,-0.4);
\draw[thick,bend left=90]  (a) edge (b);
\draw[thick,bend right=90]  (c) edge (d);
\draw[thick,bend left=60] (a) edge (c);
\draw[thick,bend left=60] (b) edge (d);
\draw[thick,bend right=60] (a) edge (c);
\draw[thick,bend right=60] (b) edge (d);
\draw[fill] (0,0.4) circle (1.75pt);
\draw[fill] (3,0.4) circle (1.75pt);
\draw[fill] (0,-0.4) circle (1.75pt);
\draw[fill] (3,-0.4) circle (1.75pt);
\end{tikzpicture}
\qquad\qquad
\begin{tikzpicture}[yscale=0.7]
\draw[fill] (5,-0.4) circle (1.75pt);
\draw[thick,bend left=60] (5,-0.4) edge (5,0.4);
\draw[fill] (5,0.4) circle (1.75pt);
\draw[thick,bend left=90] (5,0.4) edge (8.5,0.4);
\draw[fill] (8.5,0.4) circle (1.75pt);
\draw[thick,bend right=60] (8.5,0.4) edge (8.5,-0.4);
\draw[fill] (8.5,-0.4) circle (1.75pt);
\draw[fill] (5.5,-0.4) circle (1.75pt);
\draw[thick,bend right=60] (5.5,-0.4) edge (5.5,0.4);
\draw[fill] (5.5,0.4) circle (1.75pt);
\draw[thick,bend right=90] (5.5,-0.4) edge (9,-0.4);
\draw[fill] (9,-0.4) circle (1.75pt);
\draw[thick,bend right=60] (9,-0.4) edge (9,0.4);
\draw[fill] (9,0.4) circle (1.75pt);
\draw[thick,dashed] (5.25,-0.7) -- (5.25,0.7);
\draw[thick,dashed] (8.75,-0.7) -- (8.75,0.7);
\end{tikzpicture}
\caption{The graph on the left admits a rigid two-partition into equal intervals (right); thus there is equality in \eqref{eq:nopt-lower-bound}. We will return to this graph in Example~\ref{ex:nomon}.}
\label{fig:nomon}
\end{figure}

\begin{figure}[H]
\begin{tikzpicture}
\coordinate (a) at (0,0);
\coordinate (b) at (1,0);
\coordinate (c) at (2.5,0);
\coordinate (d) at (3.5,0);
\draw[fill] (b) circle (1.75pt);
\draw[fill] (c) circle (1.75pt);
\draw[thick,bend left=90] (a) edge (b);
\draw[thick,bend right=90] (a) edge (b);
\draw[thick,bend left=90] (c) edge (d);
\draw[thick,bend right=90] (c) edge (d);
\draw[thick] (b) -- (c);
\end{tikzpicture}
\qquad\qquad
\begin{tikzpicture}
\coordinate (a) at (0,0);
\coordinate (b1) at (1,0);
\coordinate (b2) at (0.67,0.33);
\coordinate (c1) at (1.75,0);
\coordinate (c2) at (2.25,0);
\coordinate (d1) at (3,0);
\coordinate (d2) at (3.33,0.33);
\coordinate (e) at (4,0);
\draw[fill] (b1) circle (1.75pt);
\draw[fill] (b2) circle (1.75pt);
\draw[fill] (c1) circle (1.75pt);
\draw[fill] (c2) circle (1.75pt);
\draw[fill] (d1) circle (1.75pt);
\draw[fill] (d2) circle (1.75pt);
\draw[thick,bend left=75] (a) edge (b2);
\draw[thick,bend right=90] (a) edge (b1);
\draw[thick] (b1) edge (c1);
\draw[thick] (c2) edge (d1);
\draw[thick,bend left=90] (e) edge (d1);
\draw[thick,bend right=75] (e) edge (d2);
\draw[thick,dashed] (2,0.5) -- (2,-0.5);
\draw[thick,dashed] (0.5,0) -- (1.25,0.35);
\draw[thick,dashed] (3.5,0) -- (2.75,0.35);
\end{tikzpicture}
\caption{The dumbbell graph on the left admits a (non-rigid only) two-partition into equal intervals (right); thus there is equality in the second inequality in \eqref{eq:nopt-lower-bound}, but the first inequality is strict. Observe that this graph contains an Eulerian path.}
\label{fig:dumbdumb}
\end{figure}
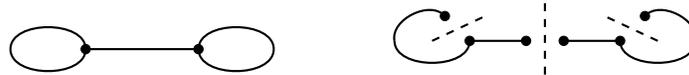

See also \cite[Section~7]{KenKurLen20}, where the graphs of Figures \ref{fig:nomon} and \ref{fig:dumbdumb} are considered. Lemma 7.1 of \cite{KenKurLen20} provides a complement to Theorem~\ref{th:nopt-lower-bound}: if, for $p=\infty$, there is a $k$-partition $\parti$ of a graph $\Graph$ whose energy $\nenergy[\infty] (\parti)$ equals $\pi^2k^2/L^2$, then this partition is a minimiser realising $\noptenergylax[k,\infty] (\Graph)$, and in particular the minimal energy also equals $\pi^2k^2/L^2$.

\begin{proof}[Proof of Theorem~\ref{th:nopt-lower-bound}]
Fix $k\geq 1$. We give the proof for $\noptenergy[k,p]$, since the argument for $\noptenergylax[k,p]$ is identical (note that due to the statement about equality the statement for $\noptenergylax[k,p] (\Graph)$ does not imply the full statement for $\noptenergy[k,p] (\Graph)$). 
As in the proof of Theorem~\ref{th:dopt-general-lower-bound}, by monotonicity in $p$ it suffices to prove the inequality for $p=1$.
 To this end, we suppose that $\Graph_1,\ldots,\Graph_k$ are the clusters of an optimal partition associated with $\noptenergy[k,1] (\Graph)$, then
\begin{equation}
\label{eq:cluster-length-sum}
	|\Graph_1|+\ldots +|\Graph_k| = L.
\end{equation}
Applying Proposition~\ref{thm:nicaise}.(1) to each cluster, we have $\mu_2 (\Graph_i) \geq \pi^2/|\Graph_i|^2$ for all $i=1,\ldots,k$ and so
\begin{displaymath}
	\noptenergy[k,1](\Graph)=\frac{1}{k}\sum_{i=1}^k \mu_2 (\Graph_i)\geq \pi^2\left(\frac{1}{k}\sum_{i=1}^k \frac{1}{|\Graph_i|^2}\right)
	\geq \pi^2\left(\frac{1}{k}\sum_{i=1}^k |\Graph_i|\right)^{-2} = \frac{\pi^2 k^2}{L^2},
\end{displaymath}
where we have applied \eqref{eq:cluster-length-sum} and, as usual, Jensen's inequality.

Equality in \eqref{eq:nopt-lower-bound} implies in particular that there is an optimising partition $\{\Graph_1,\ldots,\Graph_k\}$ yielding equality in the application of Proposition~\ref{thm:nicaise}.(1) and Jensen's inequality. This, in turn, requires that the cluster $\Graph_i$ is an interval of length $L/k$, for every $i=1,\ldots,k$.
\end{proof}

\begin{remark}
Unlike in the Dirichlet case, the condition for equality in the lower bound~\eqref{eq:nopt-lower-bound} does not prevent the graph from being doubly connected. In other words, we cannot expect an improved version of~\eqref{eq:nopt-lower-bound} for general doubly connected $\Graph$. A simple example is given by the loop, for which $\noptenergy[k,p] (\Graph)=\noptenergylax[k,p] (\Graph)=\doptenergy[k,p] (\Graph)=\frac{\pi^2 k^2}{L^2}$ for all $k$ and all $p$.
\end{remark}

We complement Theorem~\ref{th:nopt-lower-bound} with some sufficient conditions for equality which are easy to check.

\begin{proposition}
\label{prop:nopt-lower-bound-equality}
Suppose that the compact and connected graph $\Graph$ has an Eulerian path.
\begin{enumerate}
\item For all $p \in [1,\infty]$ and all $k \geq 1$ there is equality $\noptenergylax[k,p] (\Graph) = \frac{\pi^2k^2}{L^2}$ in \eqref{eq:nopt-lower-bound}.
\item If, in addition, for given $k\geq 2$ the girth $\mathfrak{s} \in (0,\infty]$ of $\Graph$ satisfies $\mathfrak{s} \geq L/k$, then also $\noptenergy[k,p] (\Graph) = \frac{\pi^2k^2}{L^2}$ for all $p \in [1,\infty]$.
\end{enumerate}
\end{proposition}

For graphs without an Eulerian path, it is still possible for there to be equality for at least some values of $k$, as the next proposition shows (the graph of Figure~\ref{fig:nomon} also provides an example). It seems reasonable to expect that the equality $\noptenergylax[k,p] (\Graph) = \frac{\pi^2k^2}{L^2}$ or $\noptenergy[k,p] (\Graph) = \frac{\pi^2k^2}{L^2}$ for \emph{all} $k\geq 1$ implies that the graph $\Graph$ has an Eulerian path, but we will not explore this question here.

\begin{proof}
Suppose that $\Graph$ has an Eulerian path. In light of \eqref{eq:nopt-lower-bound} and the monotonicity of the optimal energies in $p$, it suffices to show that under the respective claimed conditions
\begin{displaymath}
	\noptenergylax[k,\infty](\Graph),\,\noptenergy[k,\infty](\Graph) \leq \frac{\pi^2k^2}{L^2}.
\end{displaymath}
To this end, for $\noptenergylax[k,\infty](\Graph)$ we may easily construct a test $k$-partition of $\Graph$ having energy exactly $\pi^2k^2/L^2$ by cutting the graph along its Eulerian path to create $k$ intervals of length $L/k$ each. For $\noptenergy[k,\infty](\Graph)$, we observe that this resulting partition is rigid if $L/k \leq \mathfrak{s}$, since then each cluster may self-intersect at most at its endpoint, which since $k\geq 2$ and $\Graph$ is connected is necessarily a boundary point.
\end{proof}

We finish this section with a complement to the previous proposition, which states that for every graph $\Graph$ with rationally dependent edge lengths there is a sequence of values $k$ for which there is equality $\noptenergylax[k,p] (\Graph) = \noptenergy[k,p] (\Graph) = \frac{\pi^2k^2}{L^2}$.

\begin{proposition}\label{prop:nopt-subsequence}
	Assume that the edge lengths in \(\Graph\) are pairwise rationally dependent, that is, for every pair of edges \(\me_1, \me_2\in\EdgeSet\) the quotient \({|\me_1|}/{|\me_2|}\) is rational. Then there exists some positive integer \(m\geq 1\) such that
\begin{equation}
\label{eq:nopt-subsequence}
	\noptenergylax[jm,p] (\Graph) = \noptenergy[jm,p] (\Graph) = \frac{\pi^2(jm)^2}{L^2}
\end{equation}
for any integer \(j\geq 1\) and any \(p\in [1,\infty]\).
\end{proposition}

\begin{proof}
As \(\noptenergy[k,p] (\Graph)\geq \noptenergylax[k,p] (\Graph)\) both satisfy \eqref{eq:nopt-lower-bound} and are monotonically decreasing in \(p\in [1,\infty]\) for any \(k\geq 1\), it suffices to prove existence of some integer \(m\geq 1\) with
\begin{displaymath}
	\noptenergy[jm,\infty] (\Graph) \leq \frac{\pi^2(jm)^2}{L^2}
\end{displaymath}
for all \(j\geq 1\). 
First, we observe that the edge lengths are pairwise rationally dependent if and only if there is some positive real number \(r>0\) such that \(m_\me:={|\me|}/{s}\) is an integer for all edges \(\me\in\EdgeSet\). We set
\begin{displaymath}
	m:=\sum_{e\in\EdgeSet} m_e =\frac{L}{r}.
\end{displaymath}
For \(j\geq 1\) let \(\parti\) be the rigid \(jm\)-partition obtained after cutting through every vertex of \(\Graph\) and then dividing each edge \(e\in \EdgeSet\) into \(jm_e\) intervals of equal length \({s}/{j}\), so \(\parti\) is an equipartition with
\begin{displaymath}
	\noptenergy[jm,\infty] (\mathcal{G})\leq \nenergy[\infty] (\parti)=\frac{\pi^2j^2}{r^2}=\frac{\pi^2(jm)^2}{L^2}.
\end{displaymath}
This proves the claim.
\end{proof}

\begin{remark}\label{rem:nopt-subsequence}
	In particular, the previous proposition holds for equilateral graphs, and the proof shows that in this case we may choose \(m\) as the cardinality of the edge set in that case.
\end{remark}

\section{Upper bounds}
\label{sec:estimates-minimal-upper}

\subsection{Dirichlet partitions}
\label{sec:dirichlet-upper}
We next consider upper bounds on $\doptenergy[k,p](\Graph)$.
\begin{theorem}
\label{thm:dopt-upper-bound}
	Suppose $\Graph$ is a compact and connected metric graph. Then we have
	\begin{equation}
	\label{eq:dopt-upper-bound}
		\doptenergy[k,p](\Graph)\leq \frac{\pi^2}{L^2}\left(k+\left(|\EdgeSet|-1-\left\lfloor\frac{|\NeuSet|}{2}\right\rfloor\right)\right)^2
	\end{equation}
	for all sufficiently large integers \(k\geq 2\) and all \(p\in [1,\infty]\), where \(|\NeuSet|\) denotes the number vertices in \(\Graph\) of degree \(1\). In particular, \eqref{eq:dopt-upper-bound} holds whenever
	\begin{displaymath}
		k \geq \frac{L}{\ell_{\min}} + |\EdgeSet| - 1,
	\end{displaymath}
where we recall that $\ell_{\min} = \min_{\me \in \EdgeSet} |\me|$ is the minimal edge length.
\end{theorem}

\begin{proof}
By monotonicity, it suffices to prove the theorem for $p=\infty$. The proof consists of constructing a ``test partition'' formed by dividing each edge into a given number of intervals in accordance with its length, where the lengths are suitably chosen.

Without loss of generality, we may assume that \(\Graph\) has at least two edges, otherwise \(\mathcal G\) would be a cycle or an interval and in both cases \eqref{eq:dopt-upper-bound} is obviously satisfied. Let \(\EdgeSet_\NeuSet\) denote the set of pendant edges in \(\EdgeSet\), i.e\, those edges containing a vertex of degree one. Note that, since \(\Graph\) has at least two edges and \(\Graph\) is connected, each edge contains at most one vertex of degree one, and thus \(|\EdgeSet_\NeuSet|=|\NeuSet|\) holds. Fix an integer $n \geq 1$ large enough, so that $\frac{L}{n} \leq |\me|$ for all $\me\in\EdgeSet$. Now for each $\me\in\EdgeSet$ there exists an integer $m_\me$ such that
\begin{equation}\label{eq:dopt-choice-mj1}
	m_\me \cdot \frac{L}{n} \leq |\me| < (m_\me+1) \frac{L}{n},
\end{equation}
if \(\me\in\EdgeSet\setminus\EdgeSet_\NeuSet\) and
\begin{equation}\label{eq:dopt-choice-mj2}
	\frac{2m_\me-1}{2} \cdot \frac{L}{n} \leq |\me| < \frac{2m_\me+1}{2}\cdot \frac{L}{n}
\end{equation}
if \(\me\in\EdgeSet_\NeuSet\).  For \(\me\in\EdgeSet\setminus\EdgeSet_\NeuSet\) we then partition \(\me\) into \(m_\me\) intervals of equal length \(\frac{|\me|}{m_\me}\), and for \(\me\in\EdgeSet_\NeuSet\) we partition \(\me\) into \(m_\me\) intervals, so that the interval containing the vertex of degree one has length \(\frac{|\me|}{2m_\me+1}\) and the remaining intervals have length \(\frac{2|\me|}{2m_\me+1}\). Note that the interval lengths here are chosen so that the first Dirichlet eigenvalue of the longer intervals and the first mixed Dirichlet--Neumann eigenvalue of the shorter intervals are both equal to \(\frac{\pi^2(2m_\me+1)^2}{4|\me|^2}\). Let \(\parti\) be the \(m\)-partition thus obtained, where
\begin{displaymath}
	m:= \sum_{\me\in\EdgeSet} m_\me.
\end{displaymath}
Summing up \eqref{eq:dopt-choice-mj1} and \eqref{eq:dopt-choice-mj2} and using \(m=\sum_{\me\in\EdgeSet} m_\me\) and \(L=\sum_{\me\in\EdgeSet} |\me|\), we immediately obtain
\begin{equation}
	m-\Big\lfloor\frac{|\NeuSet|}{2}\Big\rfloor \leq n \leq m + |\EdgeSet| - 1-\Big\lfloor\frac{|\NeuSet|}{2}\Big\rfloor .
\end{equation}
By choice of the interval lengths we have
\begin{equation*}
	\denergy[\infty](\parti) \leq \max\left(\max_{1\le j\le |\NeuSet|} \frac{\pi^2(2m_j+1)^2}{4L_j^2},\max_{|\NeuSet|+1\le j\le |\EdgeSet|} \frac{\pi^2 m_j^2}{L_j^2}\right) \leq \frac{\pi^2 n^2}{L^2},
\end{equation*}
and thus $\doptenergy[m,\infty] (\Graph) \leq \frac{\pi^2 n^2}{L^2}$. Since $m \geq n - |\EdgeSet| + 1+\big\lfloor\frac{|\NeuSet|}{2}\big\rfloor$ and $\doptenergy[k,\infty](\Graph)$ is monotonically increasing in $k$ by \cite[Remark~4.11]{KenKurLen20}, we thus have
\begin{displaymath}
	\doptenergy[n-|\EdgeSet|+1+\lfloor\frac{|\NeuSet|}{2}\rfloor,\infty](\Graph) \leq \frac{\pi^2 n^2}{L^2}.
\end{displaymath}
Setting $k:= n + |\EdgeSet| - 1-\lfloor\frac{|\NeuSet|}{2}\rfloor$ in the above inequality yields \eqref{eq:dopt-upper-bound}.
\end{proof}

\begin{remark}
\label{rem:mu-k-bounds-via-ld}
It is known that  $\doptenergy[k,p](\Graph)$ dominates the $k$-th lowest eigenvalue $\mu_k$ of the Lapacian with natural vertex conditions, cf.~\cite[Prop.~5.5]{KenKurLen20}. Hence, in particular, Theorem~\ref{thm:dopt-upper-bound} yields, for sufficiently large $k$,
\[
		\mu_k\leq \frac{\pi^2}{L^2}\left(k-1+|\EdgeSet|-\left\lfloor\frac{|\NeuSet|}{2}\right\rfloor\right)^2.
	\]
	This estimate can be compared with the upper bound obtained in~\cite[Thm.~4.9]{BerKenKur17}, which in the present case of Laplacians with no Dirichlet boundary conditions reads
\[
		\mu_k\leq \frac{\pi^2}{L^2}\left(k-\frac12 +\frac{3}{2}|\EdgeSet|-\frac{3}{2}|\VertexSet|+\frac{|\NeuSet|}{2}\right)^2;
	\]
studying the class of graphs $\mathcal W_{m,n}$ (see Example~\ref{exa:pc2}), the latter bound was shown to be asymptotically sharp in~\cite[Theorem~2]{KurSer18}.
\end{remark}

\subsection{Neumann partitions}
\label{sec:neumann-upper}
Our main upper bound in this case reads as follows.

\begin{theorem}\label{thm:nopt-upper-bound}
	Suppose there exists an \(n\)-partition of \(\Graph\) such that every associated cluster \(\Graph_j\) has an Eulerian path,
then we have
	\begin{equation}\label{eq:nopt-upper-bound}	
		\noptenergylax[k,p](\Graph) \leq \noptenergy[k,p](\Graph) \leq \frac{\pi^2}{L^2}\big(k+(n-1)\big)^2
	\end{equation}
	for all sufficiently large integers \(k\geq 1\) and all \(p\in [1,\infty]\). Concretely, we may take $k \geq \max\{4|\EdgeSet|+n-1, \frac{3L}{\mathfrak{2s}} \}$, where $|\EdgeSet|$ is the number of edges of $\Graph$ and $\mathfrak{s} \in (0,\infty]$ its girth.
\end{theorem}

\begin{remark}
\label{rem:nopt-upper-bound}
	Obviously we may always choose \(n\) to be the number of edges of \(\Graph\) in {Theorem \ref{thm:nopt-upper-bound}}, leading to the bound
\begin{displaymath}
	\noptenergylax[k,p](\Graph) \leq \noptenergy[k,p](\Graph) \leq \frac{\pi^2}{L^2}\big(k+(|\EdgeSet|-1)\big)^2.
\end{displaymath}
This is valid for all $k \geq 5|E|-1$, as an inspection of the proof shows that $\mathfrak{s}$ may be replaced by the quantity $\max \mathfrak{s} (\Graph_j)$, where \(\mathfrak{s}(\Graph_j)\) is the girth of \(\Graph_j\), which in the case of each $\Graph_j$ being an edge is simply $\infty$. (We still expect this bound on $k$, like the one in Theorem~\ref{thm:nopt-upper-bound}, to be far from optimal in general.)
\end{remark}

\begin{remark}
Theorem~\ref{thm:nopt-upper-bound} can also be used to obtain a different bound on $\mu_k$ (and $\doptenergy[k,\infty]$), cf.\ Remark~\ref{rem:mu-k-bounds-via-ld}, when combined with the interlacing inequalities obtained in \cite{HofKen21}: there it is shown, using Theorem~\ref{thm:nopt-upper-bound}, that in fact
\begin{displaymath}
	\mu_k (\Graph) \leq \doptenergy[k,\infty](\Graph) \leq \frac{\pi^2}{L^2}(k+n+\beta-2)^2
\end{displaymath}
for all $k \geq \max \{n+1-\beta,1\}$.
\end{remark}

\begin{lemma}\label{lem:nopt-subpartitions}
	Given an \(n\)-partition of \(\Graph\) with associated clusters \(\Graph_1,\ldots,\Graph_n\) we have
	\begin{equation}
		\noptenergy[m,p] (\Graph) \leq\left\{ \begin{array}{ll} \left({\displaystyle \sum_{j=1}^n\frac{m_j}{m}}\noptenergy[m_j,p] (\Graph_j)^p\right)^{1/p} & \text{if } 1\leq p<\infty,\\[.3cm]
		{\displaystyle \max_{j=1,\ldots,k}}\noptenergy[m_j,\infty] (\Graph_j) & \text{if } p=\infty\end{array},\right.
	\end{equation}
	for integers \(m_j\geq 1\) and \(m=\sum_{j=1}^n m_j\). An analogous statement holds for $\noptenergylax[m,p](\Graph)$.
\end{lemma}
\begin{proof}
	We restrict ourselves to the case \(1\leq p<\infty\) and rigid partitions, since the other cases can be dealt with analogously. For each \(j\) we choose an optimal rigid \(m_j\)-partition \(\parti_j\) of \(\Graph_j\) associated with \(\noptenergy[m_j,p](\Graph_j)\) with clusters \(\Graph_j^i\) for \(i=1,\ldots,m_j\). We consider the induced rigid \(m\)-partition \(\parti\) of \(\Graph\) given by
	\begin{displaymath}
		\parti:=\bigcup_{j=1}^n\parti_j.
	\end{displaymath}
	 By optimality of \(\parti_j\) we have
	 \begin{displaymath}
	 	m_j\noptenergy[m_j,p](\Graph_j)^p=\sum_{i=1}^{m_j}\mu_2(\Graph_j^i)^p.
	 \end{displaymath}
	 Thus, we obtain
	 \begin{displaymath}
	 	\noptenergy[m,p] (\Graph) \leq \nenergy[p](\parti)
	 	=\left(\frac{1}{m}\sum_{j=1}^n\sum_{i=1}^{m_j}\mu_2(\Graph_j^i)^p\right)^{1/p}
	 	=\left(\sum_{j=1}^n\frac{m_j}{m}\noptenergy[m_j,p] (\Graph_j)^p\right)^{1/p}.
	 \end{displaymath}
	 This concludes the proof.
\end{proof}

\begin{proof}[Proof of Theorem~\ref{thm:nopt-upper-bound}]
	Again, we may restrict ourselves to $\noptenergy[k,p](\Graph)$ and the case \(p=\infty\). Similarly to the proof of Theorem \ref{thm:dopt-upper-bound}, we construct a test partition dividing each Eulerian path into intervals of equal length.  Let \(k\geq n\) be an arbitrary, sufficiently large integer with \(\frac{L}{k}\leq |\mathcal G_j|\) for \(j=1,\ldots,n\). For \(j=1,\ldots,n\) there exists an integer  \(m_j\geq 2\), so that
	\begin{equation}\label{eq:choice-mj}
		m_j\cdot \frac{L}{k}\leq |\mathcal G_j|<(m_j+1)\frac{L}{k}.
	\end{equation}
	We set \(m:=\sum_{j=1}^n m_j\). As in the proof of Theorem~\ref{thm:dopt-upper-bound}, it is immediate that
	\begin{equation}\label{eq:comparing-m-and-n}
		m\leq k\leq m+n-1.
	\end{equation}
	Since \(\Graph_j\) has an Eulerian path and every cycle in $\Graph_j$ has length at least
	\begin{displaymath}
		\mathfrak{s} \geq \frac{3L}{2k} \geq \frac{m_j+1}{m_j}\cdot\frac{L}{k} \geq \frac{|\mathcal G_j|}{m_j}
	\end{displaymath}
	(if it has any cycles at all), 
	we may apply the result of Proposition~\ref{prop:nopt-lower-bound-equality} to obtain
	\begin{equation*}
		\noptenergy[m_j,\infty](\Graph_j)=\frac{\pi^2m_j^2}{|\mathcal G_j|^2}.
	\end{equation*}
	Thus, Lemma \ref{lem:nopt-subpartitions}, the previous equality and {\eqref{eq:choice-mj}} yield
	\begin{displaymath}
		\noptenergy[m,\infty](\Graph)
		\leq \max_{j=1,\ldots,k} \noptenergy[m_j,\infty] (\Graph_j)
		= \max_{j=1,\ldots,k} \frac{\pi^2 m_j^2}{|\mathcal G_j|^2}
		\leq \frac{\pi^2k^2}{L^2}.
	\end{displaymath}
	Since \(\noptenergy[m,\infty](\Graph)\) is monotonically increasing in \(m\) for sufficiently large \(m\), in particular for $m \geq 4|\EdgeSet|$ (see \cite[Proposition~4.14]{KenKurLen20} and its proof, and note that under the assumption $k\geq 4|\EdgeSet|+n-1$, by \eqref{eq:comparing-m-and-n} we also have $m \geq 4|\EdgeSet|$), we may use \eqref{eq:comparing-m-and-n} to conclude
	\begin{displaymath}
		\noptenergy[k-n+1,\infty](\Graph)\leq \noptenergy[m,\infty](\Graph)\leq\frac{\pi^2k^2}{L^2}.
	\end{displaymath}
	Finally, replacing \(k\) by \(k+n-1\) we obtain
	\begin{displaymath}
		\noptenergy[k,\infty](\Graph)\leq\frac{\pi^2(k+n-1)^2}{L^2}=\frac{\pi^2k^2}{L^2}+\frac{2\pi^2(n-1)k}{L^2}+\frac{\pi^2(n-1)^2}{L^2}.
	\end{displaymath}
	This concludes the proof.
\end{proof}

\section{Asymptotic behaviour of the optimal partitions}
\label{sec:asymp-behaviour}

In this section we give the proof of Theorem~\ref{thm:asymptotic-size}, which establishes that the maximal cluster size of any optimal partition tends to zero as $k \to \infty$; this relies on the asymptotic behaviour of the optimal energies obtained in the previous sections. We will also give a couple of consequences of this result, as it in turn allows us to refine and sharpen certain statements from the previous sections.

\begin{proof}[Proof of Theorem~\ref{thm:asymptotic-size}]
We first give the proof in the Dirichlet case. Notationally, for any $k\geq 1$ and any $p \in [1,\infty]$ we suppose $\parti^\ast_{k,p} = \{\Graph_1, \ldots, \Graph_k \}$ to be any admissible $k$-partition realising $\doptenergy[k,p](\Graph)$. Fix $p \in [1,\infty]$. As noted in the proof of Theorem~\ref{thm:improveddirichletestimate}, there are at most $|\NeuSet|+|\PendTwoConn|$ clusters of $\parti^\ast_{k,p}$ which contain either a vertex of degree $1$ or a doubly connected pendant of $\Graph$. Denote by $j_k \leq |\NeuSet| + |\PendTwoConn| + 1$ the number of such clusters of $\parti^\ast_{k,p}$, \emph{plus} any cluster of maximal size if there is not already at least one such cluster among them, and suppose without loss of generality that these clusters are numbered $1,\ldots,j_k$. Finally, denote by $L_k$ the total length of these $j_k$ clusters; then by construction $L_{\max}^D(k) \leq L_k$. We will prove that in fact $L_k = \mathcal{O} (k^{-1})$ as $k \to \infty$.

Firstly, observe that
\begin{equation}
\label{eq:dopt-squeeze}
	\denergy[1] (\parti^\ast_{k,p}) = \frac{\pi^2}{L^2} k^2 + \mathcal{O} (k) \qquad \text{as } k \to \infty,
\end{equation}
since by monotonicity in $p$
\begin{displaymath}
	\doptenergy[k,p] (\Graph) = \denergy[p] (\parti^\ast_{k,p}) \geq \denergy[1] (\parti^\ast_{k,p}) \geq \doptenergy[k,1] (\Graph)
\end{displaymath}
and both $\doptenergy[k,p] (\Graph)$ and $\doptenergy[k,1] (\Graph)$ behave like $\frac{\pi^2}{L^2}k^2 + \mathcal{O} (k)$ as $k \to \infty$, by Theorem~\ref{thm:dirichlet-asymptotics-first-term}. Now, with the notation described above, for $k > j_k$, using that $\lambda_1 (\Graph_i) \geq \frac{\pi^2}{4|\Graph_i|^2}$ for all $i=1,\ldots,j_k$ and $\lambda_1 (\Graph_i) \geq \frac{\pi^2}{|\Graph_i|^2}$ for all $i=j_k+1,\ldots,k$, the usual argument (see \eqref{eq:estimates-to-prove-improved-lower-dopt-bound}) yields
\begin{equation*}
	\denergy[1] (\parti^\ast_{k,p}) \geq \frac{\pi^2}{4}\,\frac{\pi^2k^2}{L^2} + \frac{3\pi^2}{4}\, \frac{(k-j_k)^3}{k(L-L_k)^2}
\end{equation*}
for all $k > j_k$. Suppose now that $L_k \neq \mathcal{O}(k^{-1})$, so that, possibly up to a subsequence, $\lim_{k\to\infty} kL_k = \infty$. We consider the asymptotic behaviour of this subsequence of $k$; our goal is to show that in the asymptotic limit this expression must be larger than allowed by \eqref{eq:dopt-squeeze}. Since $j_k$ remains bounded, the first term in the above estimate converges to zero, and so is certainly of order $\mathcal{O}(1)$, while
\begin{equation*}
	\frac{(k-j_k)^3}{k(L-L_k)^2} = \frac{k^2}{(L-L_k)^2} + \mathcal{O}(k) \quad \text{as } k \to \infty.
\end{equation*}
But since
\begin{equation*}
	\frac{k^2}{(L-L_k)^2} = \frac{k^2}{L^2}\frac{1}{(1-\frac{L_k}{L})^2}
	= \frac{k^2}{L^2}\left(1+\frac{2}{L}L_k + \mathcal{O}(L_k^2)\right) \quad \text{as } k \to \infty
\end{equation*}
and $\lim_{k\to\infty} kL_k = \infty$ by assumption, this means that
\begin{displaymath}
	\denergy[1] (\parti^\ast_{k,p}) \neq \frac{\pi^2}{L^2}k^2 + \mathcal{O}(k) \quad \text{as } k \to \infty,
\end{displaymath}
a contradiction to \eqref{eq:dopt-squeeze}.

In the natural cases, the argument is similar but simpler owing to the better estimate $\mu_2(\Graph_i) \geq \frac{\pi^2}{|\Graph_i|^2}$ for all $i$. We consider $L_k := L_{\max}^{N,r}(k)$; the case ${L}_{\max}^{N,c} (k)$ is identical. We fix $p \in [1,\infty]$ and take $\parti^\ast_{k,p} = \{\Graph_1, \ldots, \Graph_k\}$ to be an optimal $k$-partition realising $\noptenergy[k,p](\Graph)$ and suppose that the cluster $\Graph_1$ has size $|\Graph_1|=L_{\max}^{N,r}(k)$. As in the Dirichlet case, due to the asymptotics \eqref{eq:neumann-weyl} of Theorem~\ref{thm:neumann-asymptotics-first-term} we have
\begin{equation}
\label{eq:neumann-squeeze}
	\nenergy[1] (\parti^\ast_{k,p}) = \frac{\pi^2}{L^2} k^2 + \mathcal{O} (k) \qquad \text{as } k \to \infty.
\end{equation}
On the other hand, for $k \geq 2$,
\begin{displaymath}
\begin{split}
	\nenergy[1] (\parti^\ast_{k,p}) &\geq \pi^2 \left( \frac{1}{k}|\Graph_1|^2 + \frac{k-1}{k}\left(\frac{1}{k-1}\sum_{i=2}^k |\Graph_i|^{-2}\right)\right)\\
	&\geq \frac{\pi^2}{kL_k} + \pi^2 \frac{(k-1)^3}{k(L-L_k)^2}.
\end{split}
\end{displaymath}
Under the assumption that $L_k \neq \mathcal{O}(k^{-1})$, the same argument as in the Dirichlet case now yields that, possibly up to a subsequence, $\nenergy[1](\parti^\ast_{k,p}) \neq \frac{\pi^2}{L^2}k^2 + \mathcal{O}(k)$ as $k \to \infty$, contradicting \eqref{eq:neumann-squeeze}.
\end{proof}

As a first corollary of Theorem~\ref{thm:asymptotic-size} we obtain an improved version of the lower bound in Theorem~\ref{thm:dirichlet-asymptotics-first-term} for sufficiently large $k$; namely, we can drop the term $\beta$ appearing there.

\begin{corollary}\label{cor:improveddirichletestimate-largek}
	Let $\Graph$ be a compact and connected metric graph with total length $L>0$ and $|\NeuSet|$ vertices of degree one. Fix $p\in [1,\infty]$. Then there exists \(k_0\geq 2\) such that for all $k\ge k_0$ we have
	\begin{equation}
		\label{eq:improveddirichletestimate-largek}
	\doptenergy[k,p] (\Graph) \ge \frac{\pi^2}{4k L^2} \left ( k^3 + 3  (k-|\NeuSet|)^3 \right ).
	\end{equation}
\end{corollary}
\begin{proof}
	By monotonicity it is sufficient to prove the assertion for \(p=1\). For \(k\geq 2\), we suppose that \(\parti_k^D\) is an admissible \(k\)-partition  realising \(\doptenergy[k,1](\Graph)\) and \(L_{\max}^D(k)\) is the maximum length of the clusters in \(\parti_k^D\). By Theorem \ref{eq:asymptotic-size} we find some \(k_0\geq 2\) such that
		\[L_{\max}^D(k)< \ell_{\min} \] 
	holds for all \(k\geq k_0\). In particular, the clusters appearing in \(\parti_k^D\) are either intervals or stars, where all non-centre vertices are cut points. Let \(\mathcal G_1,\ldots,\mathcal G_{|\NeuSet|}\) be the clusters of \(\parti_k^D\) that contain the vertices of \(\mathcal G\) of degree one and let \(\mathcal G_{|\NeuSet|+1},\ldots,\mathcal G_k\) be the remaining clusters. We then have \(\lambda_1(\mathcal G_j)=\frac{\pi^2}{4|\mathcal G_j|^2}\) for \(j=1,\ldots,|\NeuSet|\) and \(\lambda_1(\mathcal G_j)\geq \frac{\pi^2}{|\mathcal G_j|^2}\) for \(j=|\NeuSet|+1,\ldots,k\) by \eqref{eq:nicaise-2}. Adapting the arguments in \eqref{eq:estimates-to-prove-improved-lower-dopt-bound} we obtain
		\[\doptenergy[k,1](\Graph)=\denergy[1](\parti_k^D)\geq \frac{\pi^2}{4k L^2} \left ( k^3 + 3  (k-|\NeuSet|)^3 \right ).\]
\end{proof}

As a second consequence of Theorem~\ref{thm:asymptotic-size} we will prove that, for fixed $p \in [1,\infty]$, $\noptenergy[k,p]$ is a monotonically increasing function of $k$, at least for $k$ sufficiently large.

\begin{theorem}
\label{thm:neumann-monotonicity}
Let $\Graph$ be a compact and connected graph, and fix $p \in [1,\infty]$. Then there exists $k_0\geq 2$ depending only on $\Graph$ and $p$ such that
\begin{displaymath}
	\noptenergy[k_2,p] (\Graph) \geq \noptenergy[k_1,p] (\Graph) \qquad \text{for all } k_2 \geq k_1 \geq k_0.
\end{displaymath}
\end{theorem}

We recall that the monotonicity in the connected case, $\noptenergylax[k_2,p] (\Graph) \geq \noptenergylax[k_1,p] (\Graph)$ for all $k_2 \geq k_1 \geq 1$, was already established in Remark~4.11 of \cite{KenKurLen20}, as was Theorem~\ref{thm:neumann-monotonicity} in the special case $p=\infty$ in \cite[Proposition~4.14]{KenKurLen20} (which was also required in one of the above proofs). In general we cannot necessarily expect $k_0 = 1$, see Example~\ref{ex:nomon}.

\begin{proof}
Since the case $p=\infty$ was treated in \cite{KenKurLen20}, we give the proof for $p \in [1,\infty)$. So fix $p \in [1,\infty)$ and for $k\geq 1$ denote by $\parti^\ast_{k,p} = \{\Graph_1,\ldots,\Graph_k\}$ any rigid $k$-partition achieving $\noptenergy[k,p](\Graph)$. By Theorem~\ref{thm:asymptotic-size} there exists some $k_0 = k_0 (\Graph,p)$ such that for every $k \geq k_0$ every cluster of $\parti_{k,p}^\ast$ has length strictly shorter than the shortest edge length of $\Graph$, and in particular every cluster is a tree, which meets any neighbouring cluster of $\parti_{k,p}^\ast$ at a single vertex.

It clearly suffices to prove the theorem for $k_2=k_1+1$. Fix $k \geq k_0+1$ and consider $\parti^\ast_{k,p}$; we suppose without loss of generality that
\begin{equation}
\label{eq:maxed-out}
	\mu_2 (\Graph_k) = \max_{i=1,\ldots,k} \mu_2 (\Graph_i)
\end{equation}
and that $\Graph_{k-1}$ is a neighbour of $\Graph_k$. We now set $\widetilde \Graph_{k-1} := \Graph_{k-1} \cup \Graph_k$; then since $\Graph_{k-1}$ and $\Graph_k$ necessarily meet at a single point, by \cite[Theorem~3.10(1)]{BerKenKur19}, we have $\mu_2 (\widetilde \Graph_{k-1}) \leq \mu_2 (\Graph_{k-1})$. We construct a test $k-1$-partition $\widetilde\parti := \{ \Graph_1, \ldots, \Graph_{k-2}, \widetilde \Graph_{k-1} \}$ of $\Graph$; then, again using the fact that $\Graph_{k-1}$ and $\Graph_k$ meet at a single point and $\parti_{k,p}^\ast$ was assumed rigid, $\widetilde\parti$ is a rigid $k-1$-partition of $\Graph$.

We claim that $\nenergy[p] (\parti^\ast_{k,p}) \geq \nenergy [p] (\widetilde\parti)$, from which the conclusion of the theorem in the case $p \in [1,\infty)$ will immediately follow. In fact, this is an elementary calculation using \eqref{eq:maxed-out}: it follows from \eqref{eq:maxed-out} that
\begin{displaymath}
	\mu_2 (\Graph_k)^p \geq \frac{1}{k-1}\sum_{i=1}^{k-1} \mu_2 (\Graph_i)^p,
\end{displaymath}
and hence
\begin{displaymath}
\begin{aligned}
	\nenergy[p](\parti^\ast_{k,p})^p - \nenergy[p](\widetilde{\parti})^p &= \frac{1}{k} \sum_{i=1}^k \mu_2 (\Graph_i)^p - \frac{1}{k-1}\left(\sum_{i=1}^{k-2}\mu_2 (\Graph_i)^p
		+ \mu_2 (\widetilde{\Graph}_{k-1})^p\right)\\
	&=\frac{1}{k}\mu_2(\Graph_k)^p - \frac{1}{k(k-1)}\sum_{i=1}^{k-2} \mu_2 (\Graph_i)^p + \frac{1}{k}\mu_2(\Graph_{k-1})^p - \frac{1}{k-1}\mu_2(\widetilde{\Graph}_{k-1})^p\\
	&\geq \frac{1}{k} \mu_2 (\Graph_k)^p - \frac{1}{k(k-1)} \sum_{i=1}^{k-1} \mu_2 (\Graph_i)^p
\end{aligned}
\end{displaymath}
since $\mu_2 (\Graph_{k-1}) \geq \mu_2 (\widetilde{\Graph}_{k-1})$. By \eqref{eq:maxed-out}, this latter expression is nonnegative, and so we conclude that $\nenergy[p] (\parti^\ast_{k,p}) \geq \nenergy [p] (\widetilde\parti)$, as desired.
\end{proof}

\begin{example}
\label{ex:nomon}
We consider the graph $\Graph$ depicted in Figure~\ref{fig:nomon}, which in turn was taken from \cite[Example~8.2]{KenKurLen20}; we claim that for this graph $\noptenergy[2,p] (\Graph) < \noptenergy[1,p] (\Graph)$ for all $p \in [1,\infty]$, that is, monotonicity in Theorem~\ref{thm:neumann-monotonicity} fails when $k_1 = 1$ and $k_2 = 2$.

Suppose that $\Graph$ has total length $L$ and fix $p \in [1,\infty]$. It was already shown in \cite[Example~8.2]{KenKurLen20} that $\noptenergy[2,p] (\Graph) = \frac{4\pi^2}{L^2}$. Next, we note that by definition $\noptenergy[1,p] (\Graph) = \mu_2 (\Graph)$. Now by the Band--L\'evy inequality, Proposition~\ref{thm:nicaise}(2), since $\Graph$ is not a $2$-regular pumpkin chain, we have $\mu_2 (\Graph) > \frac{4\pi^2}{L^2}$. This proves the claimed reverse monotonicity.
\end{example}

\section{Asymptotics on two simple graphs}
\label{sec:asymp}

In the previous sections, we proved that the minimal energies \(\ndoptenergy[k,p](\mathcal G)\) satisfy the Weyl-type asymptotic law
\[
	\ndoptenergy[k,p](\mathcal G)=\frac{\pi^2}{L^2}k^2+\mathcal{O}(k) \quad \text{as } k\rightarrow \infty.
\]
In this section we are going to discuss the behaviour of the first order term \(\mathcal O(k)\) in this expansion. A natural question to ask is if there exists some \(c\in \mathbb R\) such that
\[
	\ndoptenergy[k,p](\mathcal G)=\frac{\pi^2}{L^2}k^2+ck + \mathcal{O}(1) \quad \text{as } k\rightarrow \infty
\]
holds. We are going to show that in general such \(c\) does not exist. More precisely, we study the sequence given by
	\[c_k:=\frac{\ndoptenergy[k,p](\mathcal G) - \frac{\pi^2 k^2}{L^2}}{k}, \quad k\in\mathbb N\]
and give examples where \((c_k)_k\) has \(a\) limit points for some given \(a\in\mathbb N\) (equilateral star graphs with \(2a\) edges) or uncountably many limit points (two disjoint path graphs with rationally independent lengths). For simplicity of our discussion, we restrict ourselves to the case \(p=\infty\), but note that our techniques may easily be adapted to the case \(p\in [1,\infty)\).

\subsection{Equilateral stars}
\label{sec:equilateral}
For \(m\geq 3\), we consider the equilateral \(m\)-star \(\mathcal \mathcal{S}_m\) of total length \(L\).
\begin{lemma}\label{lem:doptenergy-star}
	For \(j\in\mathbb N_0\) we have
	\begin{align*}
		\doptenergy[jm+1,\infty](\mathcal{S}_m) & = \mu_{jm+1}(\mathcal{S}_m) = \frac{\pi^2m^2j^2}{L^2}, \\
		\doptenergy[jm+r,\infty](\mathcal{S}_m) & = \mu_{jm+r}(\mathcal{S}_m) = \frac{\pi^2m^2(j+\frac{1}{2})^2}{L^2},\quad r=2,\ldots,m.
	\end{align*}		
\end{lemma}
\begin{proof}
	The ordered eigenvalues \(\mu_k(\mathcal{S}_m)\) of the equilateral \(m\)-star \(\mathcal{S}_m\) are
\begin{align}\label{eq:eigenvalues-star}
	\mu_{jm+1}(\mathcal{S}_m)&=\frac{\pi^2m^2j^2}{L^2} & \mu_{jm+r}(\mathcal{S}_m)&=\frac{\pi^2m^2(j+\frac{1}{2})^2}{L^2},\quad r=2,\ldots,m
\end{align}
for \(j\in\mathbb N_0\) (cf.~\cite[Example 3]{Fri05}). By \cite[Proposition 5.5]{KenKurLen20} we have \(\mu_k(\mathcal{S}_m)\leq \doptenergy[k,\infty](\mathcal{S}_m)\) for \(k\in \mathbb N\).   Therefore it will be sufficient to find respective partitions of \(\mathcal{S}_m\) whose energies coincides with the eigenvalues in \eqref{eq:eigenvalues-star} and these partitions will be optimal.\\
For \(k=jm+1\) we consider the partition \(\mathcal P\) consisting of an equilateral \(m\)-star with edge length \(\frac{L}{2mj}\), \(m\) intervals of length \(\frac{L}{2mj}\) each having one Dirichlet and one Neumann vertex and \(m(j-1)\) intervals of length \(\frac{L}{mj}\) each having two Dirichlet vertices. Then each cluster of \(\mathcal P\) has the same Dirichlet energy \(\frac{\pi^2m^2j^2}{L^2}\) and we conclude
	\[\doptenergy[k,\infty](\mathcal{S}_m)=\denergy[p](\mathcal P)=\frac{\pi^2m^2j^2}{L^2}.\]
For \(k=mj+r\) with \(1<r\leq m\) we consider a partition \(\mathcal P\) obtained after cutting through the center vertex of the star, where the first \(r\) edges \(\me_1,\ldots,\me_r\) are divided into \(j+1\) intervals -- one of length \(\frac{L}{m(2j+1)}\) with one Neumann and one Dirichlet vertex and the other \(j\) of length \(\frac{2L}{m(2j+1)}\) with two Dirichlet vertices -- and the remaining \(m-r\) edges \(\me_{r+1},\ldots,\me_{m}\) are divided into \(j\) intervals -- one of length \(\frac{L}{m(2j-1)}\) with one Neumann and one Dirichlet vertex and the other \(j\) of length \(\frac{2L}{m(2j-1)}\) with two Dirichlet vertices. The Dirichlet energy of the clusters in \(\me_1,\ldots,\me_r\) is \(\frac{\pi^2m^2(j+\frac{1}{2})^2}{L^2}\) whereas the Dirichlet energy of the clusters in \(\me_{r+1},\ldots,\me_{m}\) is \(\frac{\pi^2m^2(j-\frac{1}{2})^2}{L^2}\). We obtain
	\[\doptenergy[k,\infty](\mathcal{S}_m)=\denergy[\infty](\mathcal P)=\frac{\pi^2m^2(j+\frac{1}{2})^2}{L^2}.\]
This concludes the proof.
\end{proof}

\begin{center}
\begin{figure}[ht]
\begin{tikzpicture}[scale=1.2]
\coordinate (c) at (0,0);
\foreach \i in {1,2,3} {
\coordinate (v\i) at (120*\i:0.5);
\coordinate (u\i) at (120*\i:0.7);
\coordinate (w\i) at (120*\i:1.7);
\coordinate (x\i) at (120*\i:1.9);
\coordinate (y\i) at (120*\i:2.4);
\coordinate (f\i) at ($(120*\i:.6)+({120*\i+90}:.3)$);
\coordinate (g\i) at ($(120*\i:.6)+({120*\i-90}:.3)$);
\coordinate (d\i) at ($(120*\i:1.8)+({120*\i+90}:.3)$);
\coordinate (e\i) at ($(120*\i:1.8)+({120*\i-90}:.3)$);
\draw[thick] (c) -- (v\i) (u\i)--(w\i) (x\i)--(y\i);
\draw[dashed, thick] (f\i) -- (g\i) (d\i) -- (e\i);
\draw[fill] (y\i) circle (1.75pt);
\draw[fill=white] (v\i) circle (1.75pt) (u\i) circle (1.75pt) (w\i) circle (1.75pt) (x\i) circle (1.75pt);}
\draw[fill] (c) circle (1.75pt);
\end{tikzpicture}
\begin{tikzpicture}[scale=1.2]
\coordinate (c) at (0,0);
\foreach \i in {1,2} {
\coordinate (v\i) at (120*\i:.15);
\coordinate (w\i) at (120*\i:.95);
\coordinate (u\i) at (120*\i:1.15);
\coordinate (x\i) at (120*\i:1.95);
\coordinate (y\i) at (120*\i:2.15);
\coordinate (z\i) at (120*\i:2.55);
\coordinate (h\i) at ({60+120*\i}:.4);
\coordinate (f\i) at ($(120*\i:1.05)+({120*\i+90}:.3)$);
\coordinate (g\i) at ($(120*\i:1.05)+({120*\i-90}:.3)$);
\coordinate (d\i) at ($(120*\i:2.05)+({120*\i+90}:.3)$);
\coordinate (e\i) at ($(120*\i:2.05)+({120*\i-90}:.3)$);
\draw[thick] (v\i) -- (w\i) (u\i) -- (x\i) (y\i) -- (z\i);
\draw[dashed, thick] (f\i) -- (g\i) (d\i) -- (e\i);
\draw[thick,dashed] (c) -- (h\i);
\draw[fill] (z\i) circle (1.75pt);
\draw[fill=white] (v\i) circle (1.75pt) (u\i) circle (1.75pt) (w\i) circle (1.75pt) (x\i) circle (1.75pt) (y\i) circle (1.75pt);
}
\foreach \i in {3} {
\coordinate (v\i) at (120*\i:.15);
\coordinate (w\i) at (120*\i:1.5);
\coordinate (u\i) at (120*\i:1.7);
\coordinate (x\i) at (120*\i:2.35);
\coordinate (h\i) at ({60+120*\i}:.4);
\coordinate (f\i) at ($(120*\i:1.6)+({120*\i+90}:.3)$);
\coordinate (g\i) at ($(120*\i:1.6)+({120*\i-90}:.3)$);
\draw[thick] (v\i) -- (w\i) (u\i) -- (x\i);
\draw[dashed, thick] (f\i) -- (g\i);
\draw[thick,dashed] (c) -- (h\i);
\draw[fill] (x\i) circle (1.75pt);
\draw[fill=white] (v\i) circle (1.75pt) (u\i) circle (1.75pt) (w\i) circle (1.75pt);
}
\end{tikzpicture}
\begin{tikzpicture}[scale=1.2]
\coordinate (c) at (0,0);
\foreach \i in {1,2,3} {
\coordinate (v\i) at (120*\i:.15);
\coordinate (w\i) at (120*\i:.95);
\coordinate (u\i) at (120*\i:1.15);
\coordinate (x\i) at (120*\i:1.95);
\coordinate (y\i) at (120*\i:2.15);
\coordinate (z\i) at (120*\i:2.55);
\coordinate (h\i) at ({60+120*\i}:.4);
\coordinate (f\i) at ($(120*\i:1.05)+({120*\i+90}:.3)$);
\coordinate (g\i) at ($(120*\i:1.05)+({120*\i-90}:.3)$);
\coordinate (d\i) at ($(120*\i:2.05)+({120*\i+90}:.3)$);
\coordinate (e\i) at ($(120*\i:2.05)+({120*\i-90}:.3)$);
\draw[thick] (v\i) -- (w\i) (u\i) -- (x\i) (y\i) -- (z\i);
\draw[dashed, thick] (f\i) -- (g\i) (d\i) -- (e\i);
\draw[thick,dashed] (c) -- (h\i);
\draw[fill] (z\i) circle (1.75pt);
\draw[fill=white] (v\i) circle (1.75pt) (u\i) circle (1.75pt) (w\i) circle (1.75pt) (x\i) circle (1.75pt) (y\i) circle (1.75pt);
}

\end{tikzpicture}
\vspace{.5em}
\caption{The optimal \(7\)-, \(8\)- and \(9\)-partitions of the \(3\)-star in the proof of Lemma~\ref{lem:doptenergy-star}. White vertices denote vertices with Dirichlet conditions.}\label{fig:doptenergy-3star}
\end{figure}
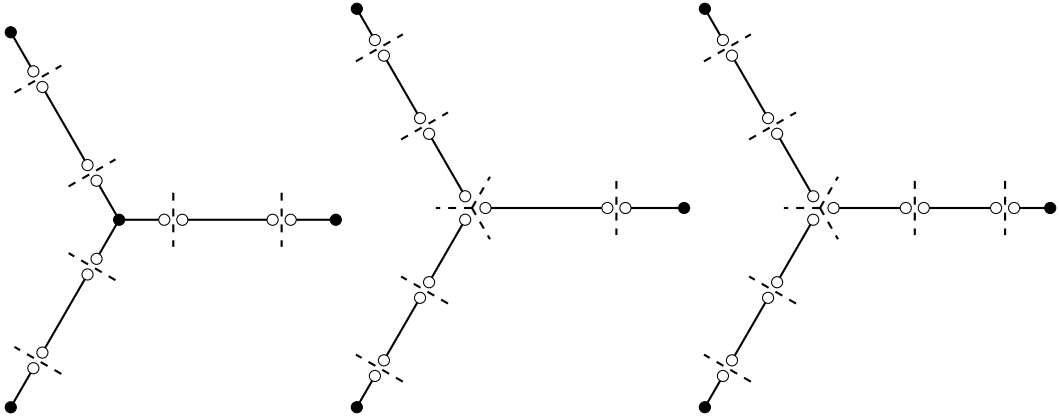
\end{center}

\begin{proposition}
\label{prop:ck-1}
	The limit set of the sequence \((c_k)_{k\in\mathbb N}\) with
			\[c_k:=\frac{\doptenergy[k,\infty](\mathcal{S}_m)-\frac{\pi^2k^2}{L^2}}{k},\quad k\in\mathbb N,\]
	is
	\begin{equation}
	\label{eq:ck-1}
		\left\{-\frac{2\pi^2}{L^2}\right\}\cup \left\{\frac{2\pi^2(s-1-\frac{m}{2})}{L^2}~\big|~s=1,\ldots,m-1\right\}.
	\end{equation}
	In particular, \((c_k)_{k\in\mathbb N}\) has \(m-1\) limit points if \(m\) is even and \(m\) limit points if \(m\) is odd.
\end{proposition}
\begin{proof}
	The assertion immediately follows from Lemma \ref{lem:doptenergy-star} if one considers the subsequences \((c_{k_j})_{j\in\mathbb N_0}\) given by \(k_j:=jm+r\) for \(r=1,\ldots,m\) and \(j\in\mathbb N_0\). Indeed, for \(r=1\), we have
	\begin{displaymath}
		k_jc_{k_j}  =\frac{\pi^2m^2j^2}{L^2}-\frac{\pi^2k_j^2}{L^2}
		 =\frac{\pi^2}{L^2}\left[(k_j-1)^2-k_j^2\right]
		 =\frac{\pi^2}{L^2}(-2k_j+1)
	\end{displaymath}
	and, thus, \(c_{k_j}\rightarrow -\frac{2\pi^2}{L^2}\) as \(k_j\rightarrow\infty\). For \(1<r\leq m\), we have
	\begin{align*}
		k_jc_{k_j} =\frac{\pi^2m^2(j+\frac{1}{2})^2}{L^2}-\frac{\pi^2k_j^2}{L^2}
		& =\frac{\pi^2}{L^2}\left[\left(k_j+\frac{m}{2}-r\right)^2-k_j^2\right] \\
		& =\frac{\pi^2}{L^2}\left[2k_j\left(\frac{m}{2}-r\right)+\left(\frac{m}{2}-r\right)^2\right]
	\end{align*}
	and, thus, \(c_{k_j}\rightarrow \frac{\pi^2(m-2r)}{L^2}\) as \(k_j\rightarrow\infty\). Note that, if \(m\) is even, the limit point in the second case coincides with the one in the first case for \(r=\frac{m}{2}+1\).
\end{proof}

\begin{remark}
\label{rem:second-weyl}
Proposition~\ref{prop:ck-1} also shows that, if we write
\begin{displaymath}
	\mu_k (\mathcal{S}_m) = \frac{\pi^2 k^2}{L^2} + c_k k,
\end{displaymath}
then the set of points of accumulation of $(c_k)_{k\in\N}$ is exactly \eqref{eq:ck-1}. This is an immediate consequence of the equality $\doptenergy[k,\infty] (\mathcal{S}_m) = \mu_k (\mathcal{S}_m)$ for all $k \geq 1$, as shown in Lemma~\ref{lem:doptenergy-star}. In particular, we have an explicit example for the non-existence of a second term in the Weyl asymptotics for $\mu_k$.
\end{remark}

We now consider the case of natural partitions.

\begin{lemma}\label{lem:noptenergy-star}
	For \(j\in\mathbb N_0\) we have
	\begin{align}\label{eq:noptenergy-mstar}
		\noptenergy[jm+r,\infty](\mathcal{S}_m) =\begin{cases}\displaystyle\frac{\pi^2m^2(j+\frac{1}{2})^2}{L^2},& r=1,\ldots,\displaystyle\big\lfloor\frac{m}{2}\big\rfloor,\\[0.2cm]
		\displaystyle\frac{\pi^2m^2(j+1)^2}{L^2},& r=\displaystyle\big\lfloor\frac{m}{2}\big\rfloor+1,\ldots,m.
		\end{cases}
	\end{align}	
	
\end{lemma}
\begin{proof}
	We set \(k=jm+r\). We first show that \(\nenergy[\infty]\) is indeed bounded from below by the terms appearing on the right-hand-side of \eqref{eq:noptenergy-mstar} respectively. For an arbitrary \(k\)-partition  \(\mathcal P\) of \(\mathcal{S}_m\), let \(\mathcal P'\) denote the set of clusters in \(\mathcal P\) that intersect at least two edges of \(\mathcal{S}_m\) and, for each edge \(\me_i\) of \(\mathcal{S}_m\), let \(\mathcal P_i\) denote the set of clusters in \(\mathcal P\) that only intersect \(\me_i\). Furthermore, let \(k'=|\mathcal P'|\) and \(k_i:=|\mathcal P_i|\). By choice of \(k'\) and \(k_i\), we have \(k=k'+\sum_{i=1}^mk_i\) and \(k'\leq \frac{m}{2}\), where the latter holds, since each edge of \(\mathcal{S}_m\) intersects at most one of the clusters in \(\mathcal P'\). All clusters in \(\mathcal P_i\) are intervals, so we may assume that each element of \(\mathcal P_i\) has the same length \(\ell_i\). (Note that we only decrease \(\nenergy[\infty]\) if we adjust the length of the single intervals, so that all off them have the same length.) In particular, we have \(\mu_2(\mathcal G_i)=\frac{\pi^2}{\ell_i^2}\) for all \(\mathcal G_i\in\mathcal P_i\).
	
	Now, let us first consider the case \(1\leq r\leq \frac{m}{2}\). Without loss of generality, we may assume that \(\ell_i>\frac{L}{m(j+\frac{1}{2})}\) holds for \(i=1,\ldots,m\) -- otherwise, \(\nenergy[\infty](\mathcal P)\geq \frac{\pi^2m^2(j+\frac{1}{2})^2}{L^2}\) would obviously be satisfied. We obtain
		\[\frac{L}{m}\geq\sum_{\mathcal G_i\in\mathcal P_i}|\mathcal G_i|=k_i\ell_i>\frac{k_iL}{m(j+\frac{1}{2})}\]
and, thus, \(k_i\leq j\) for \(i=1,\ldots m\). This, in turn, implies
	\[k'=k-\sum_{i=1}^mk_i\geq jm+r-jm=r\geq 1,\]
i.e. \(\mathcal P'\) is non-empty. We consider an arbitrary element \(\mathcal G'\in \mathcal P'\). For \(i=1,\ldots,m\) with \(|\me_i\cap\mathcal G'|>0\) we have
	\[|\me_i\cap\mathcal G'|=\frac{L}{m}-k_i\ell_i<\frac{L}{m}-\frac{jL}{m(j+\frac{1}{2})}=\frac{L}{2m(j+\frac{1}{2})}.\]
Thus, \(\mathcal G'\) is a metric star whose maximum length \(\ell_{\mathrm{max}}(\mathcal G')\) is bounded from above by \(\frac{L}{2m(j+\frac{1}{2})}\). We obtain
	\[\nenergy[\infty](\parti)\geq \mu_2(\mathcal G')\geq \frac{\pi^2}{4\ell_\mathrm{max}(\mathcal G')^2}>\frac{\pi^2m^2(j+\frac{1}{2})^2}{L^2},\]
where the second step follows from \cite[Lemma 3.3]{AmiCoh18}.\\
Next, we consider the case \(\frac{m}{2}<r\leq m\). First note that \(\nenergy[\infty](\mathcal P)\geq \frac{\pi^2m^2(j+1)^2}{L^2}\) is obviously satisfied if \(\ell_i\leq \frac{L}{m(j+1)}\) holds. On the other hand, the case \(\ell_i>\frac{L}{m(j+1)}\) for all \(i\) does not occure, since then following the argumentation of the first case yields \(k'\geq r >\frac{m}{2}\), which is a contradiction to \(k'\leq \frac{m}{2}\), as we stated in the the beginning of the proof.

Altogether, we have seen that \(\noptenergy[k,\infty](\mathcal{S}_m)\) is indeed bounded from below by the terms appearing on the right-hand-side. To show equality, we simply present \(k\)-partitions with Neumann energy equal to the right-hand-side -- obviously, these partitions are spectral minimal partitions. In the case \(1\leq r<\frac{m}{2}\), we make a choice of \(r\) pairs of edges and consider their respective unions \(\me_1\cup\me_2,\ldots, \me_{2r-1}\cup\me_{2r}\); each of these unions is an Eulerian path in \(\mathcal{S}_m\). Now let \(\mathcal P\) be the partition where each of these unions is decomposed into \(2j+1\) intervals of equal length \(\frac{L}{m(j+\frac{1}{2})}\) and every other edge \(\me_i,~i>2r\) is decomposed into \(j\) intervals of length \(\frac{L}{mj}\) (see the decomposition on the left in Figure \ref{fig:noptenergy-3star}). This partition has Neumann energy \(\nenergy[\infty](\mathcal P)=\frac{\pi^2m^2(j+\frac{1}{2})^2}{L^2}\). In the case \(\frac{m}{2}<r\leq m\), we consider the \(jm+r\)-partition that decomposes the first \(r\) edges into \(j+1\) intervals of length \(\frac{L}{m(j+1)}\) and the latter \(m-r\) edges into \(j\) intervals of length \(\frac{L}{mj}\) (see the two decompositions on the right in Figure \ref{fig:noptenergy-3star}). Again, this partition has the desired Neumann energy.
\end{proof}
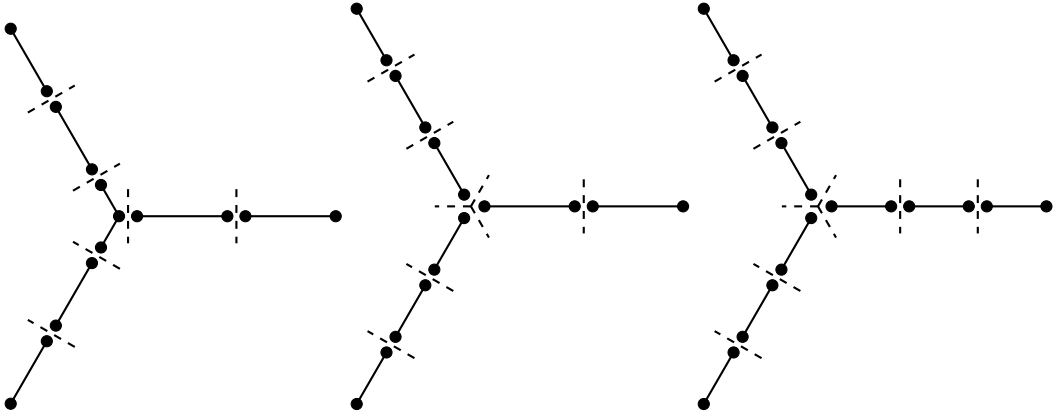
\begin{figure}[H]
\centering

\begin{tikzpicture}[scale=1.2]
\coordinate (c) at (0,0);
\foreach \i in {1,2} {
\coordinate (v\i) at (120*\i:0.4);
\coordinate (u\i) at (120*\i:0.6);
\coordinate (w\i) at (120*\i:1.4);
\coordinate (x\i) at (120*\i:1.6);
\coordinate (y\i) at (120*\i:2.4);
\coordinate (f\i) at ($(120*\i:.5)+({120*\i+90}:.3)$);
\coordinate (g\i) at ($(120*\i:.5)+({120*\i-90}:.3)$);
\coordinate (d\i) at ($(120*\i:1.5)+({120*\i+90}:.3)$);
\coordinate (e\i) at ($(120*\i:1.5)+({120*\i-90}:.3)$);
\draw[thick] (c) -- (v\i) (u\i)--(w\i) (x\i)--(y\i);
\draw[dashed, thick] (f\i) -- (g\i) (d\i) -- (e\i);
\draw[fill] (y\i) circle (1.75pt);
\draw[fill] (v\i) circle (1.75pt) (u\i) circle (1.75pt) (w\i) circle (1.75pt) (x\i) circle (1.75pt);}
\foreach \i in {3} {
\coordinate (u\i) at (120*\i:0.2);
\coordinate (w\i) at (120*\i:1.2);
\coordinate (x\i) at (120*\i:1.4);
\coordinate (y\i) at (120*\i:2.4);
\coordinate (f\i) at ($(120*\i:.1)+({120*\i+90}:.3)$);
\coordinate (g\i) at ($(120*\i:.1)+({120*\i-90}:.3)$);
\coordinate (d\i) at ($(120*\i:1.3)+({120*\i+90}:.3)$);
\coordinate (e\i) at ($(120*\i:1.3)+({120*\i-90}:.3)$);
\draw[thick] (u\i)--(w\i) (x\i)--(y\i);
\draw[dashed, thick] (f\i) -- (g\i) (d\i) -- (e\i);
\draw[fill] (y\i) circle (1.75pt);
\draw[fill] (u\i) circle (1.75pt) (w\i) circle (1.75pt) (x\i) circle (1.75pt);}
\draw[fill] (c) circle (1.75pt);
\end{tikzpicture}
\begin{tikzpicture}[scale=1.2]
\coordinate (c) at (0,0);
\foreach \i in {1,2} {
\coordinate (v\i) at (120*\i:.15);
\coordinate (w\i) at (120*\i:.81);
\coordinate (u\i) at (120*\i:1.01);
\coordinate (x\i) at (120*\i:1.67);
\coordinate (y\i) at (120*\i:1.87);
\coordinate (z\i) at (120*\i:2.53);
\coordinate (h\i) at ({60+120*\i}:.4);
\coordinate (f\i) at ($(120*\i:.91)+({120*\i+90}:.3)$);
\coordinate (g\i) at ($(120*\i:.91)+({120*\i-90}:.3)$);
\coordinate (d\i) at ($(120*\i:1.77)+({120*\i+90}:.3)$);
\coordinate (e\i) at ($(120*\i:1.77)+({120*\i-90}:.3)$);
\draw[thick] (v\i) -- (w\i) (u\i) -- (x\i) (y\i) -- (z\i);
\draw[dashed, thick] (f\i) -- (g\i) (d\i) -- (e\i);
\draw[thick,dashed] (c) -- (h\i);
\draw[fill] (z\i) circle (1.75pt);
\draw[fill] (v\i) circle (1.75pt) (u\i) circle (1.75pt) (w\i) circle (1.75pt) (x\i) circle (1.75pt) (y\i) circle (1.75pt);
}
\foreach \i in {3} {
\coordinate (v\i) at (120*\i:.15);
\coordinate (w\i) at (120*\i:1.15);
\coordinate (u\i) at (120*\i:1.35);
\coordinate (x\i) at (120*\i:2.35);
\coordinate (h\i) at ({60+120*\i}:.4);
\coordinate (f\i) at ($(120*\i:1.25)+({120*\i+90}:.3)$);
\coordinate (g\i) at ($(120*\i:1.25)+({120*\i-90}:.3)$);
\draw[thick] (v\i) -- (w\i) (u\i) -- (x\i);
\draw[dashed, thick] (f\i) -- (g\i);
\draw[thick,dashed] (c) -- (h\i);
\draw[fill] (x\i) circle (1.75pt);
\draw[fill] (v\i) circle (1.75pt) (u\i) circle (1.75pt) (w\i) circle (1.75pt);
}
\end{tikzpicture}
\begin{tikzpicture}[scale=1.2]
\coordinate (c) at (0,0);
\foreach \i in {1,2,3} {
\coordinate (v\i) at (120*\i:.15);
\coordinate (w\i) at (120*\i:.81);
\coordinate (u\i) at (120*\i:1.01);
\coordinate (x\i) at (120*\i:1.67);
\coordinate (y\i) at (120*\i:1.87);
\coordinate (z\i) at (120*\i:2.53);
\coordinate (h\i) at ({60+120*\i}:.4);
\coordinate (f\i) at ($(120*\i:.91)+({120*\i+90}:.3)$);
\coordinate (g\i) at ($(120*\i:.91)+({120*\i-90}:.3)$);
\coordinate (d\i) at ($(120*\i:1.77)+({120*\i+90}:.3)$);
\coordinate (e\i) at ($(120*\i:1.77)+({120*\i-90}:.3)$);
\draw[thick] (v\i) -- (w\i) (u\i) -- (x\i) (y\i) -- (z\i);
\draw[dashed, thick] (f\i) -- (g\i) (d\i) -- (e\i);
\draw[thick,dashed] (c) -- (h\i);
\draw[fill] (z\i) circle (1.75pt);
\draw[fill] (v\i) circle (1.75pt) (u\i) circle (1.75pt) (w\i) circle (1.75pt) (x\i) circle (1.75pt) (y\i) circle (1.75pt);
}
\end{tikzpicture}
\vspace{.5em}
\caption{The optimal \(7\)-, \(8\)- and \(9\)-partitions of the \(3\)-star in the proof of Lemma \ref{lem:noptenergy-star}.}\label{fig:noptenergy-3star}
\end{figure}
\begin{remark}
	Note that the spectral minimal partitions in the proof of Lemma \ref{lem:noptenergy-star} are not unique. For example, another optimal \(jm+1\)-partition -- whose topology differs from the one presented in the proof -- is obtained by decomposing \(\mathcal S_m\) into one equilateral \(m\)-star of total length \(\frac{L}{2j+1}\) and \(jm\) intervals of length \(\frac{L}{m(j+\frac{1}{2})}\) (see Figure \ref{fig:alternative-noptenergy-3star}). In fact, this choice seems to be more natural, since each cluster has the same Neumann energy.
\end{remark}
\begin{figure}[ht]
\centering

\begin{tikzpicture}[scale=1.2]
\coordinate (c) at (0,0);
\foreach \i in {1,2,3} {
\coordinate (v\i) at (120*\i:0.4);
\coordinate (u\i) at (120*\i:0.6);
\coordinate (w\i) at (120*\i:1.4);
\coordinate (x\i) at (120*\i:1.6);
\coordinate (y\i) at (120*\i:2.4);
\coordinate (f\i) at ($(120*\i:.5)+({120*\i+90}:.3)$);
\coordinate (g\i) at ($(120*\i:.5)+({120*\i-90}:.3)$);
\coordinate (d\i) at ($(120*\i:1.5)+({120*\i+90}:.3)$);
\coordinate (e\i) at ($(120*\i:1.5)+({120*\i-90}:.3)$);
\draw[thick] (c) -- (v\i) (u\i)--(w\i) (x\i)--(y\i);
\draw[dashed, thick] (f\i) -- (g\i) (d\i) -- (e\i);
\draw[fill] (y\i) circle (1.75pt);
\draw[fill] (v\i) circle (1.75pt) (u\i) circle (1.75pt) (w\i) circle (1.75pt) (x\i) circle (1.75pt);}
\end{tikzpicture}

\vspace{.5em}
\caption{A different optimal \(7\)-partitions of the \(3\)-star.}\label{fig:alternative-noptenergy-3star}
\end{figure}
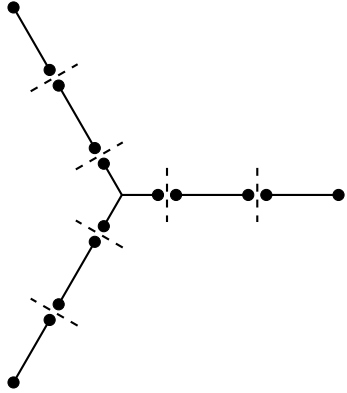
\begin{remark}
	The \(m\)-star \(\mathcal{S}_m\) can be covered with \(\frac{m}{2}\) Eulerian paths, if \(m\) is even, and \(\frac{m+1}{2}\) Eulerian paths, if \(m\) is odd. Therefore, Theorem \ref{thm:nopt-upper-bound} yields the upper bounds
		\[\noptenergy[k,\infty](\mathcal{S}_m)\leq\begin{cases}
			\displaystyle\frac{\pi^2(k+\frac{m}{2}-1)^2}{L^2}, & \text{if } m \text{ is even,}\\[.2cm]
			\displaystyle\frac{\pi^2(k+\frac{m+1}{2}-1)^2}{L^2}, & \text{if } m \text{ is odd.}
		\end{cases}\]
	Lemma \ref{lem:noptenergy-star} shows that these bounds are actually sharp if \(m\) is even and \(k=mj+1\), or \(m\) is odd and \(k=mj+\frac{m+1}{2}\) for \(j\in\N_0\) respectively.
\end{remark}
\begin{proposition}
	The limit set of the sequence \((c_k)_{k\in\mathbb N}\) with
			\[c_k:=\frac{\noptenergy[k,\infty](\mathcal{S}_m)-\frac{\pi^2k^2}{L^2}}{k},\quad k\in\mathbb N,\]
	is
		\[\{0\}\cup\left\{\frac{2\pi^2s}{L^2}~\big|~s=1,\ldots,\frac{m}{2}\right\},\]
	if \(m\) is even, and
		\[\{0\}\cup\left\{\frac{2\pi^2s}{L^2}~\big|~s=1,\ldots,\frac{m-1}{2}\right\}\cup\left\{\frac{2\pi^2(t-\frac{1}{2})}{L^2}~\big|~t=1,\ldots,\frac{m-1}{2}\right\}\]
	if \(m\) is odd.
	In particular, \((c_k)_{k\in\mathbb N}\) has \(\frac{m}{2}\) limit points if \(m\) is even and \(m\) limit points if \(m\) is odd.
\end{proposition}
\begin{proof}
	This immediately follows from Lemma \ref{lem:doptenergy-star} if one considers the subsequences \((c_{k_j})_{j\in\mathbb N_0}\) given by \(k_j:=jm+r\) for \(r=1,\ldots,m\) and \(j\in\mathbb N_0\). Indeed, calculations entirely analogous to the ones in the proof of Proposition~\ref{prop:ck-1} show that \(c_{k_j}\rightarrow \frac{\pi^2(m-2r)}{L^2}\) as \(k_j\rightarrow\infty\) for \(r=1,\ldots,\big\lfloor\frac{m}{2}\big\rfloor\), while \(c_{k_j}\rightarrow \frac{\pi^2(2m-2r)}{L^2}\) for \(r=\big\lfloor\frac{m}{2}\big\rfloor+1,\ldots,m\). Finally, we remark that if \(m\) is even, then the limit points in the two cases coincide (replace \(r\) with \(r+\frac{m}{2}\)), whereas they are distinct if \(m\) is odd.
\end{proof}

\subsection{Two disjoint intervals with rationally independent lengths}
\label{sec:irrational}

Let $\mathcal G_a= I_1 \sqcup I_a$ be the disjoint union of the intervals
\[
I_1:= [0,1], \qquad I_a:=[0,a]
\]
for some $a>0$.  Recall that
\begin{equation}\label{eq:estimate-rational-ind}
\frac{\pi^2}{(a+1)^2} k^2 \le \noptenergy[k,\infty] (\mathcal G_a) \le \frac{\pi^2}{(a+1)^2}k^2 + \frac{2\pi^2}{(a+1)^2} k + \frac{\pi^2}{(a+1)^2}
\end{equation}
holds for $k\ge 2$ by Theorem~\ref{thm:nopt-upper-bound}. As before, we are interested in 
the set of points of accumulation of the sequence \((c_k)_{k\geq 2}\) given by
\begin{equation}\label{eq:defck}
c_k=\frac{\noptenergy[k,\infty](\mathcal G_a) - \frac{\pi^2 k^2}{(a+1)^2}}{k}, \quad k\geq 2.
\end{equation}
First note that we have
\begin{equation}\label{eq:bounds-ck-rational-ind}
	0\leq c_k\leq \frac{2\pi^2}{(a+1)^2}
\end{equation}
for \(k\geq 2\) by \eqref{eq:estimate-rational-ind}. In fact, we will see that the limit set of \((c_k)_{k\geq 2}\) is the whole interval $[0, \frac{2\pi^2}{(a+1)^2}]$ if \(a\) is irrational. In order to show this, let us first compute the minimal energy \(\noptenergy[k,\infty](\mathcal G_a)\) for \(k\geq 2\).
Of course, for given $i\in \{1,\ldots , k-1\}$, an optimal $k$-partition of the form $\parti= (\Graph_1, \ldots, \Graph_i, \Graph_{i+1}, \ldots, \Graph_k)$  for \(\nenergy[\infty]\) with
\begin{align*}
\Graph_1, \ldots, \Graph_i& \subset I_1, & \Graph_{i+1}, \ldots, \Graph_{k}\subset I_a.
\end{align*}
is obtained by taking each cluster in \(I_1\) of equal length \(\frac{1}{i}\) and each cluster in \(I_a\) of equal length \(\frac{a}{k-i}\), that is,
\begin{equation}\label{eq:noptenergyequival}
\noptenergy[k,\infty](\mathcal G_a) = \min_{1\leq i\leq k-1} \max\left \{\pi^2 i^2, \frac{\pi^2(k-i)^2}{a^2}\right \}.
\end{equation}
Let us further investigate \eqref{eq:noptenergyequival}. One easily sees that
\begin{equation}
\max\left \{\pi^2 i^2, \dfrac{\pi^2 (k-i)^2}{a^2}\right \} = \begin{cases}
\dfrac{\pi^2 (k-i)^2}{a^2},&\quad i \le \left \lfloor \dfrac{k}{a+1}\right \rfloor \\ \ \\
\pi^2 i^2, &\quad i \ge \left \lceil \dfrac{k}{a+1}\right \rceil. 
\end{cases}
\end{equation}
In particular, we have
\begin{equation}\label{eq:noptenergyexform}
\begin{split}
\noptenergy[k,\infty](\mathcal G_a) &=\min_{1\leq i\leq k-1} \max\left \{\pi^2 i^2, \frac{\pi^2(k-1)^2}{a^2}\right \}\\
&= \min \left \{ \min_{1\le i\le \lfloor \frac{k}{a+1}\rfloor} \frac{\pi^2 (k-i)^2}{a^2}, \min_{\lceil \frac{k}{a+1}\rceil\leq i\leq k-1 } \pi^2 i^2\right \}\\
&= \min\left \{ \frac{\pi^2 \lceil \frac{a}{a+1} k\rceil^2}{a^2}, \pi^2 \left ( \left \lceil \frac{k}{a+1}\right \rceil \right )^2 \right \} .
\end{split}
\end{equation}
We can treat the asymptotics via study of the orbit of the rotation map $T_{\alpha}: \mathbb R/\mathbb Z \to [0,1)$, which is defined via
\begin{equation}
T_{\alpha} x= x+ \alpha \, \operatorname{mod} 1.
\end{equation}
It is a well-known fact that the orbits of the map $T_{\alpha}$ are dense in \([0,1]\) if and only if $\alpha \in \mathbb R \setminus \mathbb Q$ (see \cite[Theorem~3.13]{Dev89}). 
\begin{theorem}\label{thm:exampleintheend}
Let $c_k$, $k \geq 2$, be defined as in \eqref{eq:defck}. If $a\in \mathbb Q$, then $(c_k)_{k\geq 2}$ has a finite limit set; if $a\in \mathbb R\setminus \mathbb Q$, then the limit set of $(c_k)_{k\geq 2}$ is the whole interval $[0, \frac{2\pi^2}{(a+1)^2}]$.
\end{theorem}
\begin{proof}
Due to \eqref{eq:noptenergyexform}, we have
\begin{equation}\label{eq:ckrecharact}
c_k = \min\left \{ \dfrac{\pi^2 \left (\left \lceil \frac{k}{a+1} \right \rceil^2 - \frac{k^2}{(a+1)^2}\right )  }{k},  \dfrac{\pi^2 \left (\frac{1}{a}\left \lceil \frac{ak}{a+1} \right \rceil^2 - \frac{k^2}{(a+1)^2}\right )  }{k} \right \}.
\end{equation}
We compute
\begin{equation}\label{eq:longcalculationck}
\begin{split}
\dfrac{\pi^2 \left (\left \lceil \frac{k}{a+1} \right \rceil^2 - \frac{k^2}{(a+1)^2}\right )  }{k}&= \frac{\pi^2}{k} \left ( \left \lceil \frac{k}{a+1} \right \rceil - \frac{k}{a+1}\right ) \left ( \left \lceil \frac{k}{a+1} \right \rceil + \frac{k}{a+1} \right )\\
&= \left ( \left \lceil \frac{k}{a+1} \right \rceil - \frac{k}{a+1}\right ) \left ( \frac{2\pi^2}{a+1}+ \frac{\pi^2}{k} \left (\left \lceil \frac{k}{a+1} \right \rceil - \frac{k}{a+1}\right ) \right ) \\
&= T_{\frac{a}{a+1}}^k(0) \left ( \frac{2\pi^2}{a+1} + \frac{\pi^2}{k} T_{\frac{a}{a+1}}^k(0) \right )\\
&= \frac{2\pi^2}{a+1} T^k_{\frac{a}{a+1}}(0) + o(1) \, \text{as } k \to \infty
\end{split}
\end{equation}
and
\begin{equation}\label{eq:longcalculationck2}
\begin{split}
&\dfrac{\pi^2 \left (\frac{1}{a}\left \lceil \frac{ak}{a+1} \right \rceil^2 - \frac{k^2}{(a+1)^2}\right )  }{k}= \frac{\pi^2}{a^2 k} \left ( \left \lceil \frac{ak}{a+1} \right \rceil - \frac{ak}{a+1}\right ) \left ( \left \lceil \frac{ak}{a+1} \right \rceil + \frac{ak}{a+1} \right )\\
&\qquad \qquad\qquad = \left ( \left \lceil \frac{ak}{a+1} \right \rceil - \frac{ak}{a+1}\right ) \left ( \frac{2\pi^2}{a(a+1)}+ \frac{\pi^2}{a^2 k} \left (\left \lceil \frac{k}{a+1} \right \rceil - \frac{k}{a+1}\right ) \right ) \\
&\qquad \qquad\qquad = T_{\frac{1}{a+1}}^k(0) \left ( \frac{2\pi^2}{a(a+1)} + \frac{\pi^2}{a^2k} T_{\frac{1}{a+1}}^k(0) \right )\\
&\qquad \qquad\qquad= \frac{2\pi^2}{a(a+1)} T^{k}_{\frac{1}{a+1}}(0)+ o(1) \, \text{as } k \to \infty.
\end{split}
\end{equation}
Since the orbits of $T_{\frac{1}{a+1}}$ and $T_{\frac{a}{a+1}}$ are periodic if and only if $a\in \mathbb Q$, we deduce that $a$ has a finite limit set if and only if $a\in \mathbb Q$.  Suppose $a\in \mathbb R\setminus \mathbb Q$, then
\begin{equation}
\begin{split}
T_{\frac{1}{a+1}}^k(0)+ T_{\frac{a}{a+1}}^k(0)&=\frac{k}{a+1} - \left \lfloor \frac{k}{a+1} \right \rfloor + \frac{ak}{a+1} - \left \lfloor \frac{ak}{a+1}\right \rfloor \\
&= k- \left \lfloor \frac{k}{a+1} \right \rfloor  - \left \lfloor \frac{ak}{a+1}\right \rfloor= \left \lceil \frac{ak}{a+1} \right \rceil - \left \lfloor \frac{ak}{a+1}\right \rfloor=1.
\end{split}
\end{equation}
Let $x\in[0,1]$. Suppose $(k_n)_{n\in\mathbb N}$ is a strictly increasing sequence with
\begin{equation}
\lim_{n\to \infty} T^{k_n}_{\frac{1}{a+1}}(0) = x,
\end{equation}
then with \eqref{eq:ckrecharact}, \eqref{eq:longcalculationck} and \eqref{eq:longcalculationck2} for all $k\in \mathbb N$ we infer  
\begin{equation}
\begin{split}
\lim_{n\to \infty} c_{k_n} &= \min\left \{\frac{2\pi^2 (1-x)}{a+1}, \frac{2\pi^2  x}{a(a+1)}\right \}\\
&= \begin{cases}
\dfrac{2\pi^2 x}{a(a+1)}, &\quad \hbox{for }x\le \dfrac{a}{a+1},\\
\ \\ 
\dfrac{2\pi^2(1-x)}{a+1}, &\quad  \hbox{for }x> \dfrac{a}{a+1},
\end{cases}
\end{split}
\end{equation}
and hence the limit set of $(c_k)_{k\geq 2}$ is dense in $[0, \frac{2\pi^2}{(a+1)^2}]$. Since the limit set is clearly closed, we conclude that it equals $[0, \frac{2\pi^2}{(a+1)^2}]$.
\end{proof}

In the Dirichlet case, we may similarly consider the limit set of the sequence \((c_k)_{k\geq 2}\) given by
\begin{equation}\label{eq:cktildedef}
c_k = \frac{\doptenergy[k,\infty](\mathcal G_a) - \frac{\pi^2 k^2}{(1+a)^2}}{k},\quad k\geq 2.
\end{equation}
On an interval $I=[0, \ell]$ we have
\begin{equation}\label{eq:relationintervals}
\doptenergy[k+1,\infty](I)= \noptenergy[k,\infty](I),\quad k\geq 2,
\end{equation}
which directly gives us the following result.
\begin{theorem}
Let $c_k$, $k\geq 2$, be defined as in \eqref{eq:cktildedef}. If $a\in \mathbb Q$, then $(c_k)_{k\geq 2}$ has a finite limit set; if $a\in \mathbb R\setminus \mathbb Q$, then the limit set of $(c_k)_{k\geq 2}$ is the interval $[-\frac{4\pi^2}{(a+1)^2}, -\frac{2\pi^2}{(a+1)^2}]$.
\end{theorem}
\begin{proof}
Using \eqref{eq:relationintervals} yields
\begin{equation}
\doptenergy[k+2,\infty](\mathcal G_a) = \noptenergy[k, \infty](\mathcal G_a)
\end{equation}
and, thus,
\begin{equation*}
\begin{split}
	\frac{\doptenergy[k+2,\infty](\mathcal G_a) - \frac{\pi^2 (k+2)^2}{(1+a)^2}}{k+2} &=\frac{k}{k+2}\frac{\noptenergy[k,\infty](\mathcal G_a) - \frac{\pi^2 k^2}{(a+1)^2}}{k}\\
	&\qquad -4\frac{\pi^2}{(a+1)^2}+o(1)\qquad \text{as } k \to \infty.
	\end{split}
\end{equation*}
The assertion now follows immediately from Theorem~\ref{thm:exampleintheend}.
\end{proof}

\appendix

\section{Partitions of metric graphs}
\label{sec:appart}

In this appendix we collect the basic definitions and notions of partitions of metric graphs, introduced in \cite{KenKurLen20}, which are needed throughout, as well as a few basic properties. Everything in this appendix is taken or adapted from \cite{KenKurLen20}. Throughout, we assume that all our metric graphs $\Graph = (\VertexSet,\EdgeSet)$ are finite, compact and connected, {\it i.e.} the vertex set is finite, the edge set is finite, each edge has finite length and there is a continuous path (for the natural Euclidean metric) connecting any two points on the graph.

Formally, we identify each edge $e \in \EdgeSet$ with a compact interval in $\R$, and each vertex $v \in \VertexSet(\Graph)$ with a \emph{set} $V$ of endpoints of intervals, namely those intervals whose associated edges are incident with $v$.

\begin{definition}\label{def:urg}
We call any equivalence class of metric graphs with respect to the equivalence relation defined by isometric isomorphisms between graphs an \emph{ur-graph}.
\end{definition}

Any two metric graphs which differ by the presence of a finite number of vertices of degree two (dummy vertices) are representatives of the same ur-graph. Formally speaking, we work at the level of ur-graphs, although in practice we may suppress this technicality and work with any representative as convenient; in particular we take any finite set of points in the graph to be vertices as convenient.

\begin{definition}
\label{def:graph-cut-new1}
Let $\Graph,\Graph'$ be ur-graphs. Then $\Graph'$ is called a \emph{cut} of $\Graph$ if there exist a representative $\widehat\Graph$ of $\Graph$ and a representative $\widehat\Graph'$ of $\Graph'$ with vertex sets $\VertexSet(\widehat\Graph) = \{v_1,\ldots,v_{\cardV}\}$ and $\VertexSet(\widehat\Graph') = \{v_1',\ldots,v_{{\cardV}'}'\}$ and edge sets $\EdgeSet(\widehat\Graph)$ and  $\EdgeSet(\widehat\Graph')$, respectively, such that
\begin{enumerate}
\item[(a)] $\EdgeSet(\widehat\Graph) = \EdgeSet(\widehat\Graph')$,
\item[(b)] ${\cardV}'\geq {\cardV}$, and
\item[(c)] for all $n' = 1,\ldots,{{\cardV}'}$, keeping the identification of each vertex $v$ as a set $V$ of interval endpoints, we have
\begin{displaymath}
	{V}_{n'}'(\widehat\Graph') \subset V_n(\widehat\Graph)
\end{displaymath}
for some $n=1,\ldots,{\cardV}$.
\end{enumerate}
\end{definition}

\begin{definition}[Partitions of a graph]
\label{def:partition}
Let $k\geq 1$ and let $\Graph$ be an ur-graph.
\begin{enumerate}
\item We call any set of $k$ metric graphs
\[
	\parti := \{ \Graph_{1},\ldots,\Graph_{k}\}
\]
a \emph{$k$-partition} of $\Graph$ if there is a cut $\Graph' = \bigsqcup_{j=1}^{k_0}\Graph_{i_j}$ of $\Graph$, $k_0 \geq k$, such that for all $j=1,\ldots,k$ there exists some $i_j$ such that $\Graph_{i_j}=\Graph_j$, with $i_{j_1} \neq i_{j_2}$ for $j_1\neq j_2$. In this case, we refer to the components $\Graph_{1},\ldots,\Graph_{k}$ as the \emph{clusters} of the partition $\parti$ (\textit{arising} from the cut $\Graph'$).
\item If in (1) there exists a cut $\Graph'$ of $\Graph$ such that $\Graph' = \bigsqcup_{i=1}^k \Graph_i$, then we say the partition $\parti = \{ \Graph_1,\ldots, \Graph_k\}$ is \emph{exhaustive}.
\end{enumerate}
\end{definition}


\begin{definition}[Cluster supports]
\label{def:cluster-support}
Let $\Graph$ be an ur-graph and let $\parti = \{\Graph_1,\ldots,\Graph_k\}$ be a $k$-partition of $\Graph$, arising from the cut $\Graph' = \bigsqcup_{i=1}^{k_0}\Graph_i$, $k_0\geq k$, of $\Graph$. 
We identify $\Graph$ and $\Graph'$ with any respective representatives satisfying the conditions of Definition~\ref{def:graph-cut-new1}, that is, in such a way that $\EdgeSet (\Graph) = \EdgeSet (\Graph')$.
\begin{enumerate}
\item For each $i=1,\ldots,k$, we denote by $\Omega_i$ the unique closed subset of $\Graph$ such that
\begin{displaymath}
\{ e \in \EdgeSet (\Graph'): e \subset \Graph_i\}= \left\{ e \in \EdgeSet (\Graph): e \subset \Omega_i \right\}
\end{displaymath}
and call the set $\Omega_i$ the \emph{cluster support} (associated with the cluster $\Graph_i$), or just \emph{support} for short.
\item We call the set
\begin{equation}
\label{eq:partition-support}
	\Omega := \bigcup\limits_{i=1}^k \Omega_i
\end{equation}
the \emph{support} of the partition $\parti$.
\end{enumerate}
\end{definition}

In practice we may simply (but somewhat imprecisely) identify the cluster $\Graph_i$ with the subset $\Omega_i$ of $\Graph$ to which it corresponds. We always assume we have chosen a representative of the ur-graph $\Graph$ in such a way that $\partial\Omega_i \subset \VertexSet(\Graph)$ for all $i$, that is, we only cut through vertices of (the representative of) $\Graph$ in order to form the clusters $\Graph_i$.

\begin{definition}
\label{def:neighbours}
Let $\parti = \{ \Graph_1, \ldots, \Graph_k \}$ be a $k$-partition of  $\Graph$, and denote by $\Omega_1,\ldots,\Omega_k$ the respective cluster supports. We say that $\Graph_i,\Graph_j$, $i,j=1,\ldots,k$, $i \neq j$, are  \emph{neighbours}, or \emph{neighbouring clusters}, if $\partial\Omega_i \cap \partial\Omega_j \neq \emptyset$.
\end{definition}

We conclude with a subset of the classification of partitions introduced in \cite[Definition~2.12]{KenKurLen20}. In the current context we are only interested in these two types, as they are exactly the classes of partitions in which our spectral functionals admit minimisers.

\begin{definition}[Classification of partitions]
\label{def:classification}
Let $\Graph$ be an ur-graph.
\begin{enumerate}
\item Any partition $\parti = \{\Graph_1, \ldots, \Graph_k\}$ of $\Graph$ satisfying Definition~\ref{def:partition} will be called \textit{connected}, if in addition each cluster $\Graph_1,\ldots,\Graph_k$ is connected. We denote the set of all exhaustive connected $k$-partitions of $\Graph$ by $\mathfrak{C}_k (\Graph)$.
\item A connected partition $\parti$ of $\Graph$ will be called \emph{rigid} if we only cut vertices on the boundary of $\Omega_i$ to create the graph $\Graph_i$. We denote the set of all exhaustive rigid $k$-partitions of $\Graph$ by $\mathfrak{R}_k (\Graph)$.
\end{enumerate}
\end{definition}

\section{Isoperimetric inequalities}
\label{sec:appii}

The first isoperimetric inequality for metric graphs was discovered by Nicaise 35 years ago; it is sharp, as shown by Friedlander 20 years later. However, it has been observed by several authors that special classes of graphs allow for improved isometric inequalities: to help keep the paper more self-contained and since we make use of some of them repeatedly, we list here the most relevant.

\begin{proposition}
\label{thm:nicaise}
Let $\Graph$ be any compact connected metric graph. Then the following assertions hold.
\begin{enumerate}[(1)]
\item We have
\begin{equation}
\label{eq:nicaise}
	\lambda_1 (\Graph) \geq \frac{\pi^2}{4|\Graph|^2} \qquad \text{and} \qquad \mu_2 (\Graph) \geq \frac{\pi^2}{|\Graph|^2},
\end{equation}
where in the first case $\Graph$ is equipped with at least one Dirichlet vertex. Equality in either inequality implies that $\Graph$ is a path graph (interval) of length $|\Graph|$, with a Dirichlet vertex at exactly one endpoint and a natural (Neumann) condition at the other in the first case, and natural conditions at both endpoints in the second case.

\item 
If additionally (possibly upon identifying all Dirichlet vertices)  $\Graph$ is doubly connected, then we have 
\begin{equation}
\label{eq:nicaise-2}
	\lambda_1 (\Graph) \geq \frac{\pi^2}{|\Graph|^2} \qquad \text{and} \qquad \mu_2 (\Graph) \geq \frac{4\pi^2}{|\Graph|^2},
\end{equation}
In this case, equality is attained only by 2-regular pumpkin chains (second case), or 2-regular pumpkin chains with two edges of equal length attached to one of the endpoints and the degenerate case of an interval with two Dirichlet endpoints (\emph{caterpillar graphs}, first case).
\end{enumerate}
\end{proposition}

The inequalities in (1) may be found in \cite[Th\'eor\`eme~3.1]{Nic87}. For the characterisation of equality, see for example \cite[Theorem~1]{Fri05}. For the inequalities in (2) we refer to~\cite[Theorem~2.1]{BanLev17}	 and~\cite[Theorem~3.4 and Lemma~4.3]{BerKenKur17}.

\end{document}